\documentclass{article}
\usepackage[utf8]{inputenc}
\usepackage{mathrsfs}  
\usepackage[left=0.8in,top=0.8in,right=0.8in,bottom=0.8in]{geometry}

\newcommand{\be}{\begin{equation}}
\newcommand{\ee}{\end{equation}}
\newcommand{\bea}{\begin{eqnarray}}
\newcommand{\eea}{\end{eqnarray}}
\newcommand{\ba}{\begin{array}}
\newcommand{\ea}{\end{array}}
\newcommand{\rf}[1] {(\ref{#1})}
\newcommand{\half}{\mbox{$\frac{1}{2}$}}
\newcommand{\eps}{\epsilon}

\usepackage{Viscous}	

\usepackage{enumitem}
\usepackage{authblk}
\title{The Effects of Viscosity on the Linear Stability of  Damped Stokes Waves, Downshifting,  and Rogue Wave Generation}
\author[1]{A. Calini}
\author[2]{C.L. Ellisor}
\author[3]{C.M. Schober \thanks{Corresponding author: cschober@ucf.edu}}
\author[3]{E. Smith}
\affil[1]{Department of Mathematics, College of Charleston}
\affil[2]{Department of Mathematics, Wake Forest University}
\affil[3]{Department of Mathematics, University of Central Florida}

\date{\today}

\begin{document}

\def\ri{\mathrm{i}}
\def\rd{\mathrm{d}}
\def\e{\mathrm{e}}

\maketitle

\begin{abstract}
  We investigate a higher order nonlinear Schr\"odinger equation with linear damping and weak viscosity, recently proposed as a model for deep water waves exhibiting frequency downshifting. Through analysis and numerical simulations, we discuss how the viscosity affects the linear stability of the Stokes wave solution, enhances rogue wave formation, and leads to permanent downshift in the spectral peak. The novel results in this work include the analysis of the transition  from the initial Benjamin-Feir instability to a predominantly oscillatory behavior, which takes place in a time interval when most rogue wave activity occurs. In addition, we propose new criteria for downshifting in the spectral peak and determine the relation between the time of permanent downshift and the location of the global minimum of the momentum and the magnitude of its second derivative.
\end{abstract}

\section{Introduction}
The modeling of water waves has provided a rich array of interesting nonlinear phenomena. A notable example is the modulational instability of the Stokes wave in deep water, along with its attendant consequences. 
For instance,  Lake et.~al.~\cite{lake_yuen_rungaldier_ferguson_1977} conducted laboratory experiments to examine the long-time evolution of modulated Stokes waves. They observed a permanent downward shift in the dominant frequency of the wave train for very steep waves, following the growth of the modulational instability. This result was later extended to wave trains of moderate steepness in the separate studies by \cite{melville_1982} and \cite{su82}.
These experiments did not align with the theoretical predictions of recurrence given by the nonlinear Schr\"odinger (NLS) equation (see equation~\eqref{vHONLS} below with $\epsilon = \Gamma =0$); in fact, Lake noted that downshifting was inconsistent with the
assumption of a constant carrier frequency in the derivation of the NLS equation.

Following the original experiment by Lake et.~al.~\cite{lake_yuen_rungaldier_ferguson_1977}, further studies focused on
developing models which provide theoretical explanations for
frequency downshifting.
Lo and Mei~\cite{lo_mei_1985} found that only temporary
downshift was obtained  using the Dysthe equation,
derived by retaining higher order terms in the asymptotic expansion for the surface wave displacement \cite{dysthe},
primarily due to an insufficient change in the flux.  
The prevailing viewpoint is that dissipation and wind play crucial roles in facilitating downshifting.
Models which included  wind and wave breaking terms, to account for the experimentally observed asymmetric spectral evolutions, were examined in the works of \cite{eeltink2017, hara_mei_1991,TrulsenDysthe1990,SS}
and \cite{tulin99,huang96,okamura96}, respectively. 

With the inclusion of  various types of nonlinear dissipation in the Dysthe model, permanent downshift was obtained for both breaking waves \cite{TrulsenDysthe1990} and non-breaking waves \cite{KatoOikawa95}. Moreover, in a nonlinear damped  higher order NLS (NLD-HONLS) model, where  the damping is
 effective only near the crest of the envelope, permanent downshifting was shown  
to affect  rogue wave activity \cite{islas,SS}. 
In a study of how well widely used deep-water  models capture the frequency downshifting exhibited in a series  of laboratory experiments, the NLD-HONLS was found to  be the most accurate  model in typical situations~\cite{chb2019}.
Even so,  the NLD-HONLS model does not capture all downshifting phenomena, as
it was found to be ineffective in describing the frequency downshift observed in ocean swells.

In \cite{didyza2008,cartergovan}, a viscous  version  of the Dysthe system (vDysthe) is
derived from first-principles by incorporating viscosity into Euler's equations. The introduction of weak viscosity provides a mechanism for  frequency downshifting that is independent of wind or wave breaking. This
is an important feature when analyzing ocean swells, which are deep water waves no longer sustained  by wind. By comparing simulations of the viscous Dysthe system for the surface displacement with  ocean swell field data,
it was determined that the vDysthe model serves as a
reasonable model  for swells that display frequency downshift \cite{zaugcarter21}. 

In this work we study the linear stability of  damped Stokes waves, the mechanism for frequency downshifting, and rogue wave activity in the framework of
a viscous higher-order nonlinear Schr\"odinger equation: 
\begin{equation}
\label{vHONLS}
    \ri u_t + u_{xx} + 2u|u|^2 + \ri \Gamma u + \epsilon\left[ 2u \mathscr{H} (|u|^2)_x - 8\ri|u|^2 u_x + \frac{\ri}{2} u_{xxx} + 2 \nu \Gamma u_x\right] - \ri \epsilon^2 \delta u_{xx}= 0,
\end{equation}
with $\epsilon>0$ and  $0 < \delta \ll 1$. The dependent variable $u$ represents  the complex envelope of the wave train and is spatially periodic of period $L$. The expression $\displaystyle \mathscr{H}(f)(x) := \frac{1}{\pi} \int_{-\infty}^{\infty} \frac{f(\xi)}{x - \xi} d\xi$ denotes the Hilbert transform of $f$. 
The  derivation of \eqref{vHONLS} from the equations for weakly damped free-surface flows is discussed in Appendix~\ref{derive_vHONLS}.

For $\epsilon=0$, equation~\eqref{vHONLS} reduces to the linearly damped NLS, with $\ri \Gamma u$ the linear damping term.
When  $\Gamma=\delta=0$ (HONLS), the first three terms within the brackets
give a Hamiltonian version of the Dysthe equation ~\cite{fd11,gramstad_trulsen_2011}.
The effects of viscosity are modeled by the terms $\ri \Gamma u$ and $2\epsilon \Gamma u_x$.
The parameter $\nu$ (0 or 1) is introduced to allow a comparison  of the viscous HONLS equation
(vHONLS: $\nu = 1, \delta = 0$)
 with the linear damped HONLS equation (LDHONLS:  $\nu = \delta = 0$).
The vHONLS is a version of the dissipative Gramstad-Trulsen equation introduced in \cite{chb2019}. We also included a diffusive term $-\ri\epsilon^2  \delta u_{xx}$
(dvHONLS) in order to investigate how a small amount of diffusion can offset instabilities caused by the viscosity. 
The complete set of  parameter values for the models
is provided in (\ref{block}). 

 \medskip

In Section 2, we investigate the linear stability of the damped Stokes wave solution of equation~\eqref{vHONLS}. In particular, we compare the linear dynamics of a perturbation  of the damped Stokes wave in the non-viscous regime to that for the case of added viscous damping. We also discuss the stabilizing effects of diffusion through the addition a higher-order term modeling diffusion effects.  The resulting linearization leads to a linear system with time-dependent coefficients, which adds some challenges to the linear stability analysis. In the infinite-time limit, we can rigorously describe the asymptotic behavior of the perturbation. For finite time, we describe the development of an initial Benjamin-Feir instability and the transition to predominantly oscillatory dynamics; both phenomena are related to the formation and cessation of rogue waves.
\smallskip

The results for the infinite-time limit reveal, as observed by other authors \cite{cartergovan}, the destabilizing effects of the viscous terms, which cannot be entirely offset by the exponential decay due to the linear damping (see Proposition~\ref{largetime} and Corollary~\ref{ltinstabu}). In fact, the Fourier modes with high wave numbers will grow exponentially (over a long time). The linear analysis also partially explains the downshifting phenomenon, since the instability due to viscosity only appears in the modes associated with negative values of the wave number (see Corollary~\ref{negk}).  The addition of diffusion has a stabilizing effect and, if large enough (but still small), every mode of the perturbative term will decay exponentially (see Proposition~\ref{prop2p2}.)

While the infinite-time limit of the linearization provides the description of the combined effect of viscosity, linear damping and, when present, diffusion, it is  important to understand the initial behavior a solution of~\eqref{vHONLS} close to a damped Stokes wave that is unstable at $t=0$. When one or more modes of the initial perturbation satisfy a linear instability condition (the analogue of the classical Benjamin-Feir inequality), these modes exhibit an interesting dynamics: their amplitudes grow exponentially up to a critical time, at which the exponential growth begins to saturate. After that, the behavior becomes predominantly oscillatory, superimposed upon the much slower exponential growth introduced by the viscosity. 
The initial growth can be substantial, with the solution reaching its absolute maximum over a very large time interval, and it appears to be responsible for the rogue wave formation.
In Section~\ref{shortmoder}, we analyse this phenomenon with and without viscosity effects. Interestingly, the linearization~\eqref{linNLS} can be reduced to (much simpler looking) second order ordinary differential equations with time-dependent coefficients,~\eqref{2de0} and~\eqref{2deg}. For the non-viscous case, we show that the higher order terms lead to an increased ``effective amplitude" of the damped Stokes wave, thus making it more unstable and delaying the onset of the damped oscillatory behavior (see Proposition~\ref{0transition}). When viscosity is present, the analysis is more challenging since the resulting ODE has now {\em  time-dependent complex} coefficients. Nevertheless, we are able to provide a qualitative analysis of the transition between the initial exponential growth and the predominantly oscillatory behavior. We show that the viscosity has a destabilizing effect, not only for $t\rightarrow +\infty$, but also on the initial evolution, by 
increasing the initial instability and delaying the onset of the oscillatory behavior (see Theorem~\ref{gtransition}). 
This is in agreement with the numerical observation
of increased rogue wave activity  in the viscous HONLS equation as compared with the linear damped HONLS
equation on the same time scales (see Figures \ref{NRWs} (a) -- (b)).

In Section~\ref{FDS_RWA},  we examine the impact of 
viscous damping on 
frequency downshifting and rogue wave formation.
Frequency downshifting is commonly  evaluated using two distinct spectral measures: the spectral mean $k_m$, a weighted average of the wave's spectral content, and the spectral peak $k_{peak}$, labeling the Fourier mode with the highest energy. 
The two measures capture different  characteristics of the spectral  distribution and can sometimes lead to differing conclusions on downshifting; in particular,
$k_m$  does not provide insight into the occurrence or the time of a permanent downshift in $k_{peak}$. 

Since $k_m$ can be expressed as the ratio of the 
momentum $P$ to the energy $E$ of the system,  explicit conditions on $P$
have been given for the occurrence of downshifting
in the spectral mean sense, as outlined in \eqref{key}.
On the other hand, while it is well-known that $P$ measures the asymmetry of the Fourier modes, no criteria have been proposed for predicting the occurrence and  timing of
permanent downshift in $k_{peak}$ for viscous HONLS models in terms of the behavior of the momentum $P$.

 In Section~\ref{visc_results}  we 
 determine  the structure of $P$ for solutions of the vHONLS and dvHONLS equations with perturbed Stokes wave initial data.  
We derive  an analytical estimate for
the time $t_{P_{min}}$ of the global minimum of $P$, which
is in very good agreement with the numerical data.
We
find that  both $t_{Pmin}$  and  the magnitudes of the first and second derivatives of $P$
play a crucial role in the underlying mechanism that leads to permanent
downshifting. The  shifts  in $k_{peak}$ result from significant
fluctuations in the derivatives of the momentum, which
experience substantial decay  for $ 0 < t <  t_{P_{min}}$.
Permanent downshift  in $k_{peak}$ occurs once the second derivative
$P''$  has become small enough
for changes in the first derivative $P'$  to no longer be sufficient for $k_{peak}$ to upshift back.

In Section \ref{ensemble}, we discuss the results of an  ensemble of
 experiments,  where  the damping parameter $\Gamma$ is varied.
We find that permanent downshift occurs in the solutions  of
both the vHONLS and dvHONLS equations across all initial conditions in the two-unstable modes regime  and for all given values of $\Gamma$.   

Last, we discuss the effects of viscous damping and downshifting  on rogue wave formation.
Comparing rogue wave activity in two sets of experiments: a) LDHONLS versus vHONLS evolution, and b) LDHONLS versus dvHONLS evolution,
we find that, for all values of $\Gamma$, as many or more rogue waves develop in the vHONLS and dvHONLS models as compared with the LDHONLS model.
The experiments also indicate that rogue waves do not occur after permanent downshifting. 

\section{Linear Stability Analysis of the viscous damped Stokes wave}
In this section we study the linear stability of the {\em damped Stokes wave} solution
\begin{equation}
\label{dStokes}
u_a^\Gamma(t)=a \e^{-\Gamma t} \text{exp} \left( \ri a^2 \frac{1 - \e^{-2\Gamma t}}{\Gamma}\right)
\end{equation}
of equation~\eqref{vHONLS}, where $a$ is a real constant representing the amplitude at $t=0$ and $\Gamma>0$ the linear damping parameter. Since $u_a^\Gamma$ does not depend on $x$, it is  a solution of both the linearly damped NLS, obtained by setting $\epsilon=0$, and the full viscous Dysthe model~\eqref{vHONLS} with and without the diffusion term. 

Substituting the ansatz $u(x,t)=u_a^\Gamma(t)(1 +\alpha v(x,t))$, with $0<\alpha \ll 1$ into~\eqref{vHONLS}, we derive the linearization
\begin{equation}
    \label{linNLS}
\ri v_t+v_{xx}+2 |u_a^\Gamma|^2 (v+v^*)+\epsilon \left[2 |u_a^\Gamma|^2 \mathscr{H} (v_x+v^*_x)
-8\ri |u_a^\Gamma|^2v_x +\frac{\ri}{2}v_{xxx}+2\nu \Gamma v_x\right]=0.
\end{equation}

\noindent
{\bf Important Note:} 
Throughout this section we will focus on the solutions of equation~\eqref{linNLS}, where we have factored out the overall exponential decay coming from the linear damping. The quantity $v(x,t)$ describes the {\em relative deviation} of the perturbed solution from the damped Stokes wave and the {\em spatial features} of the perturbed solution. Both are important to understand in relation to rogue wave formation and frequency downshifting. We will study the long-time asymptotic behavior of $v$ as well as any transient dynamics due to an initial Benjamin-Feir instability, and discuss  the implications for the linear stability of the original damped Stokes wave.
\smallskip

The main analysis is carried out in absence of diffusion ($\delta=0$), but it is easily modify when $\delta\not=0$. In particular, we will examine the effects of diffusion on the long-term dynamics at the end of section~\ref{largetime}.

\smallskip

Since we are interested in the time behavior of the Fourier coefficients of the complex series for $v$, for $k\geq 0$, we let 
\[
v_k(x,t)=A_k(t)\e^{\ri k x}+B_k(t)\e^{-\ri k x}.
\]
Using the identity $\mathscr{H} ((v_k)_x+(v_k^*)_x)= k (v_k+v_k^*)$ and writing 
 $A_k=x_k+\ri y_k$ and $B_k=w_k + \ri z_k$, we arrive at the following linear system for $\vec{u}_k=(x_k, y_k,w_k, z_k)^T$:
\begin{equation}
\label{lsys1}
\frac{\rd \vec{u}_k(t)}{\rd t}=[M + E(t) N]\vec{u}_k,
\end{equation}
with 

\[
M=\begin{pmatrix} -\gamma & -\frac{1}{2}k^2(-2+\epsilon k)  & 0 & 0 \\
\frac{1}{2}k^2(-2+\epsilon k)   & -\gamma  & 0 & 0 \\
0 & 0  & \gamma & \frac{1}{2}k^2(2+\epsilon k)   \\
0  & 0  & -\frac{1}{2}k^2(2+\epsilon k)   & \gamma
\end{pmatrix} \quad \text{and} 
\]
\[
N=\begin{pmatrix} 0 & -\frac{1}{2}(1+5\epsilon k)  & 0 & \frac{1}{2}(1+\epsilon k) \\
\frac{1}{2}(1+5\epsilon k)   & 0  & \frac{1}{2}(1+\epsilon k) & 0 \\
0 & \frac{1}{2}(1+\epsilon k)  & 0 & -\frac{1}{2}(1-3\epsilon k) \\
\frac{1}{2}(1+\epsilon k)  & 0  & \frac{1}{2}(1-3\epsilon k)  & 0
\end{pmatrix}.
\]
where $E(t)=4a^2\e^{-2\Gamma t}$ encodes the effects of linear damping and $\gamma=2\epsilon k \Gamma \nu\geq 0$ encodes the viscosity effects. 

\smallskip

The eigenvalues of $M$ are readily computed to be the quadruplet
\[
\gamma \pm \frac{ik^2}{2} (2+\epsilon k), \quad -\gamma  \pm \frac{ik^2}{2} (2-\epsilon k).
\]

\subsection{Long-time asymptotics}
\label{largetime}

Having recast equation~\eqref{linNLS} in the form~\eqref{lsys1}, we can use some classical results to determine the long-term behavior of its solutions (i.e. the asymptotic dynamics as $t\rightarrow \infty$.) Similar conclusions were  obtained in~\cite{cartergovan} for a related, but different system.

\begin{prop}
\label{longt}
If $\gamma=0$ (non-viscous case) , then all solutions of system~\eqref{lsys1} are bounded.\\
When $\gamma\not=0$ (viscous case), there exists a  solution of linear system~\eqref{lsys1} which grows exponentially in time at rate $\gamma$.
\end{prop}
\begin{proof}
If $\gamma=0$ ($\nu=0)$ then all solutions of $\dot{\vec{x}}= M\vec{x}$ are bounded. Since $\text{tr} (M)=0$ and $\int_0^\infty \|E(t)N\| \rd t<\infty$, then all solutions of system~\eqref{lsys1} are bounded. (See R.~Bellman, Ch. 2, Theorem 6 \cite{bellman}.) \smallskip

 When $\gamma\not=0$ ($\nu=1$), then the large time behavior of the solution is governed by the eigenvalues of $M$. In particular, for every eigenvalue $\lambda_i$ of $M$, there is a solution $\vec{x}_i$ of linear system~\eqref{lsys1} satisfying 
\[
\lim_{t\rightarrow +\infty} \frac{\ln \| \vec{x}_i \|}{t}=\Re(\lambda_i)=\pm \gamma.
\]
(See R.~Bellman, Ch. 2, Theorem 7 \cite{bellman}.)
\end{proof}
\begin{rem}
It follows that the viscous term introduces long-time instabilities in the dynamics. Since $\gamma=2\epsilon k \Gamma$ (with $k>0$), the larger the wave number, the higher the rate of exponential growth. 
\end{rem}

Recalling that $u(x,t)=u_a^\Gamma(t)(1 +\alpha v(x,t))$, in order to obtain the growth rate of the perturbation of the damped Stokes wave, we multiply  $v$ by the factor $\e^{-\Gamma t}$. Then, the large time behavior of its $k$-th mode $A_k(t)\e^{-\Gamma t}\e^{\ri kx} +B_k(t)\e^{-\Gamma t} \e^{-\ri kx}$  is determined  by the sign of $\pm \gamma -\Gamma=(\pm 2\epsilon k -1)\Gamma$, as described by the following results:

\begin{cor}
\label{ltinstabu}
If $\gamma=0$, then the damped Stokes wave $u^\Gamma_a$ is linearly stable, as all solutions of the linearization of the vHONLS equation~\eqref{vHONLS} remain bounded for all times and decay exponentially as $\e^{-\Gamma t}$ as $t\rightarrow +\infty$.

For $\gamma\not=0$, a generic perturbation of the damped Stokes wave will grow exponentially in time. While the Fourier modes with low wave numbers satisfying $k<1/(2\epsilon)$ decay exponentially, for $k>1/(2\epsilon)$ the corresponding modes will grow exponentially at rate $\Gamma(2\epsilon k -1)$. Thus $u_a^\Gamma$ is linearly unstable for any non-zero value of the viscosity.
\end{cor}

Furthermore, if we make the additional change of variables $\vec{\tilde u}_k=(x_k\e^{\gamma t}, y_k\e^{\gamma t},w_k\e^{-\gamma t}, z_k\e^{-\gamma t})^T$, we obtain a system of the same form as~\eqref{lsys1} for $\vec{\tilde u}_k$ with matrices
\[
\tilde{M}=\begin{pmatrix} 0 & -\frac{1}{2}k^2(-2+\epsilon k)  & 0 & 0 \\
\frac{1}{2}k^2(-2+\epsilon k)   & 0  & 0 & 0 \\
0 & 0  & 0& \frac{1}{2}k^2(2+\epsilon k)   \\
0  & 0  & -\frac{1}{2}k^2(2+\epsilon k)   & 0
\end{pmatrix} \quad \text{and} 
\]
\[
\tilde{N}=\begin{pmatrix} 0 & -\frac{1}{2}(1+5\epsilon k)  & 0 & \frac{1}{2}(1+\epsilon k) \e^{2\gamma t} \\
\frac{1}{2}(1+5\epsilon k)   & 0  & \frac{1}{2}(1+\epsilon k)\e^{2\gamma t} & 0 \\
0 & \frac{1}{2}(1+\epsilon k)\e^{-2\gamma t}  & 0 & -\frac{1}{2}(1-3\epsilon k) \\
\frac{1}{2}(1+\epsilon k)\e^{-2\gamma t}  & 0  & \frac{1}{2}(1-3\epsilon k)  & 0
\end{pmatrix}.
\]
A similar argument as for Proposition~\ref{longt} shows that the entries of $\vec{\tilde u}_k$ are bounded as $t\rightarrow +\infty$, thus $A_k=x_k+\ri y_k$ decays exponentially at rate $-\gamma$ and $B_k=w_k+\ri z_k$ grows exponentially at rate $\gamma$. Since the latter is the coefficient of $\e^{-\ri k x}$ in the Fourier expansion of $v$, we arrive at the following
\begin{cor}
\label{negk}
The long-time instabilities introduced by the viscous term in the linearization of the vHONLS equation~\eqref{vHONLS} occur in  the negative Fourier modes associated with  sufficiently high values of the wave number $k$ as described in Corollary~\ref{ltinstabu}.
\end{cor}
\begin{rem}
This result relates to the occurrence of frequency downshifting when viscosity is present, as the amplitudes of the negative modes become dominant for large times, as observed in~\cite{cartergovan}.
\end{rem}

For fixed  $k \geq 0$, we write  $v=A(t) \e^{\ri k x}+B(t)\e^{-\ri k x}.$  Figures~\ref{ABmod551} through~\ref{ABmod552} show the time evolution of the complex moduli of $A$ and $B$ obtained from numerical integration of system~\eqref{lsys1} with initial conditions $a_k=0.2$, $b_k=\sqrt{2}/2$ $c_k=-0.1$, $d_k=(1+\sqrt{3})/2$. While the scale varies, the main qualitative features do not depend on the choice of generic initial conditions.

Figure~\ref{ABmod55nv} shows the evolution of $|A|$ and $|B|$ when viscous effects are neglected ($\gamma=0$). For this case, purely oscillatory behaviour sets in for sufficiently large time.

\begin{figure}[!ht]
  \centerline{
    \includegraphics[width=.32\textwidth]{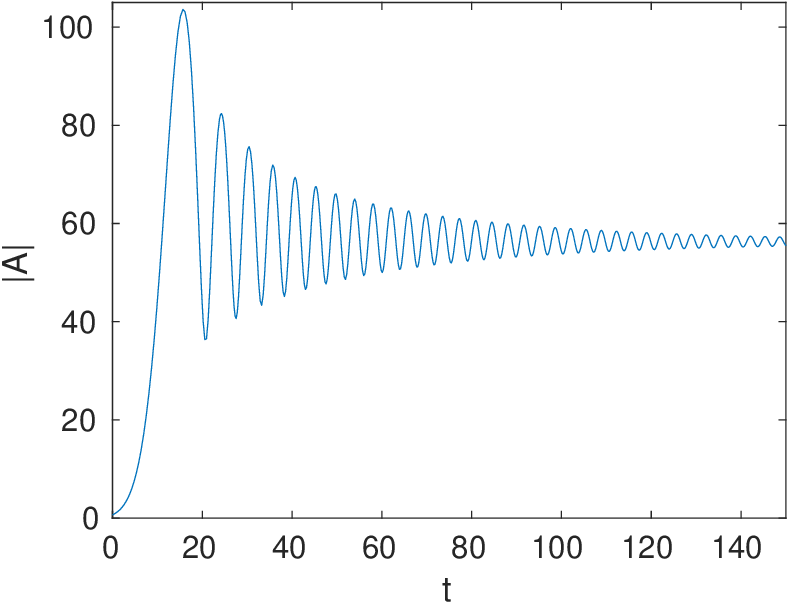}\hspace{30pt}
    \includegraphics[width=.32\textwidth]{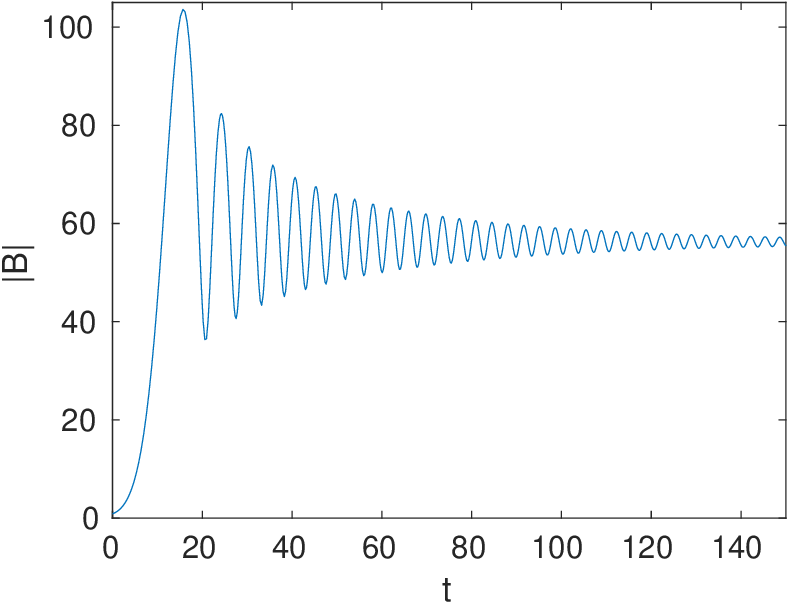} }
  \caption{Evolution of $|A|$ and $|B|$ with no viscosity for $k=1$, $a=0.55$, $\epsilon=0.05$, $\Gamma=0.01$, $\gamma=0$. }
  \label{ABmod55nv}
  \end{figure}

For the cases with non-zero viscosity (Fig.~\ref{ABmod551},~\ref{ABmod451},~\ref{ABmod552}), the long-time trend consists of oscillatory behavior superimposed to a slow decay (for $|A|$) or a slow growth (for $|B|$), both due solely to the viscous effects. (See the rightmost graphs of $|B|e^{-\gamma t}$ in figures~\ref{ABmod551},~\ref{ABmod451},~\ref{ABmod552}.)

 \begin{figure}[!ht]
  \centerline{
    \includegraphics[width=.32\textwidth]{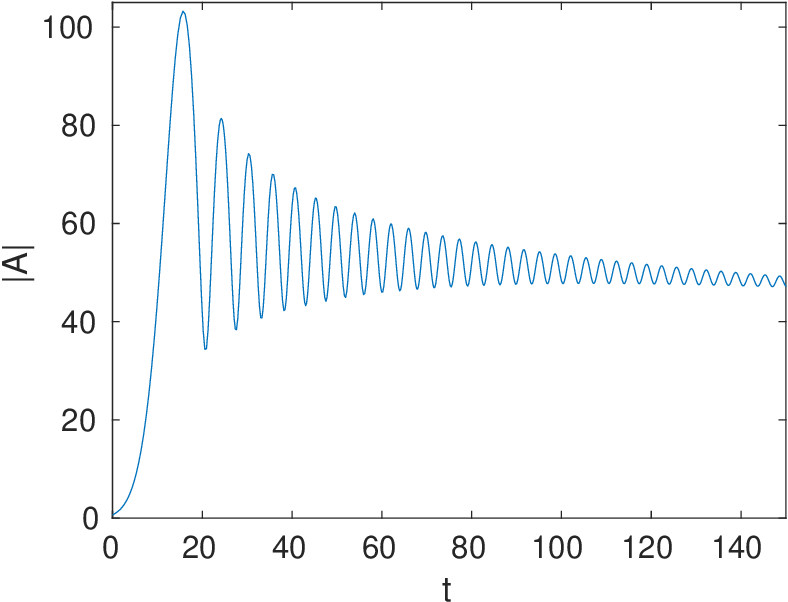}\hspace{5pt}
    \includegraphics[width=.32\textwidth]{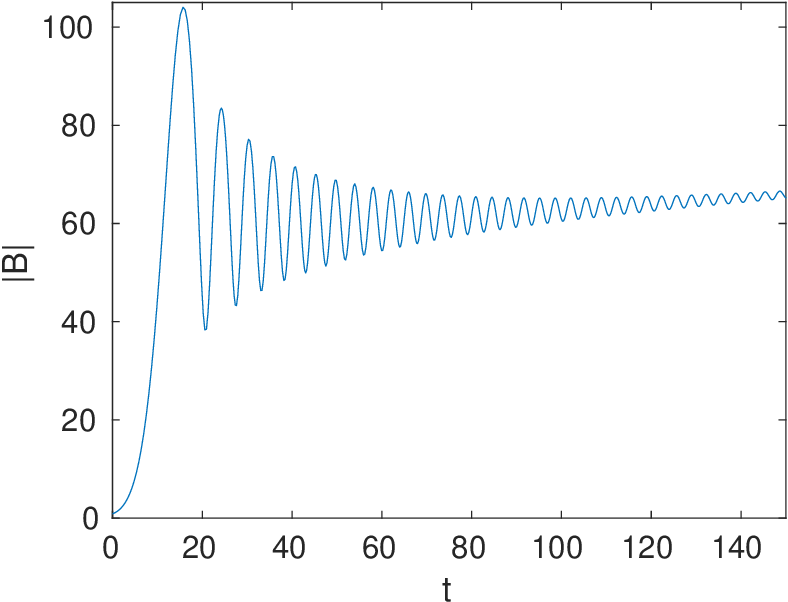}\hspace{5pt}
 \includegraphics[width=.32\textwidth]{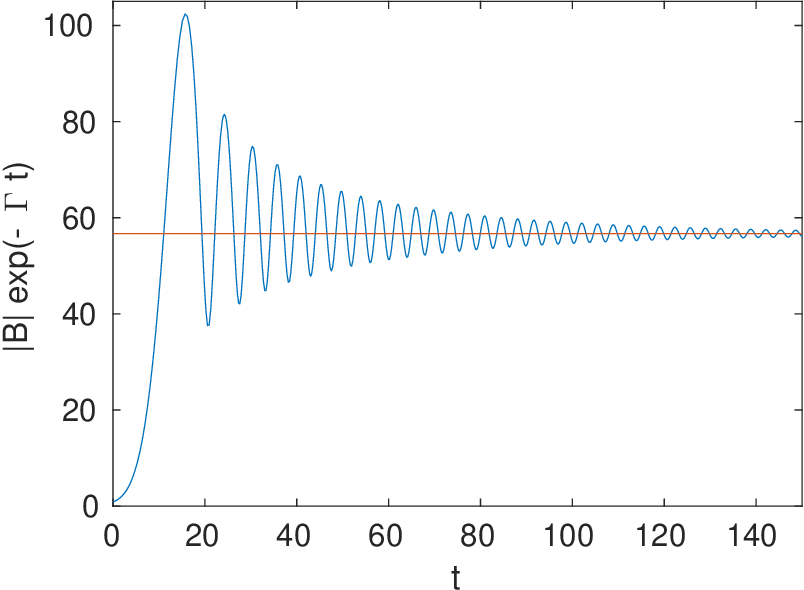}
  }
  \caption{Evolution of $|A|$, $|B|$ and $|B|\e^{-\gamma t}$ for $k=1$, $a=0.55$, $\epsilon=0.05$, $\Gamma=0.01$, $\gamma=0.001$.}
  \label{ABmod551}
  \end{figure}

\begin{figure}[!ht]
  \centerline{
    \includegraphics[width=.32\textwidth]{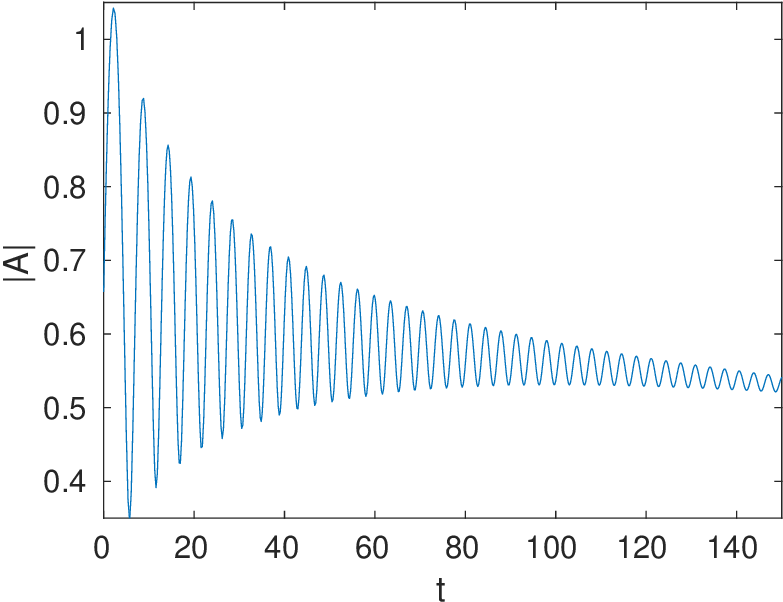}\hspace{5pt}
    \includegraphics[width=.32\textwidth]{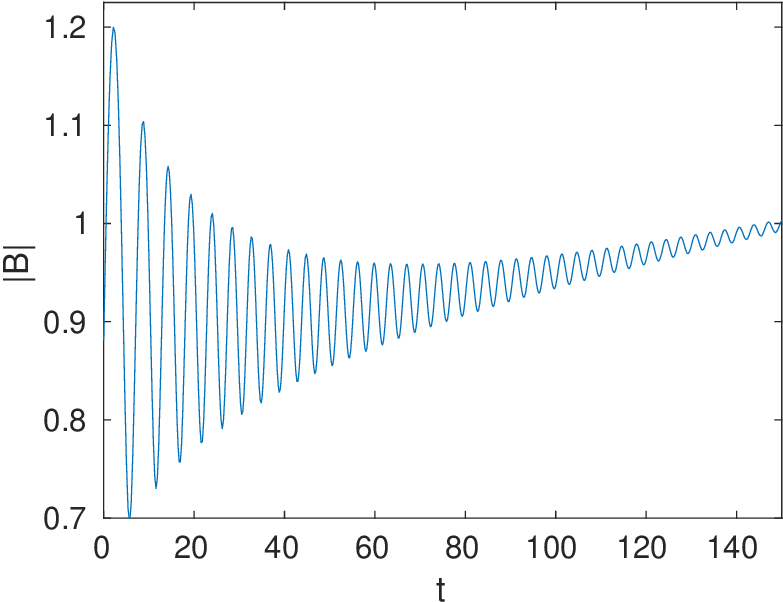}\hspace{5pt}
    \includegraphics[width=.32\textwidth]{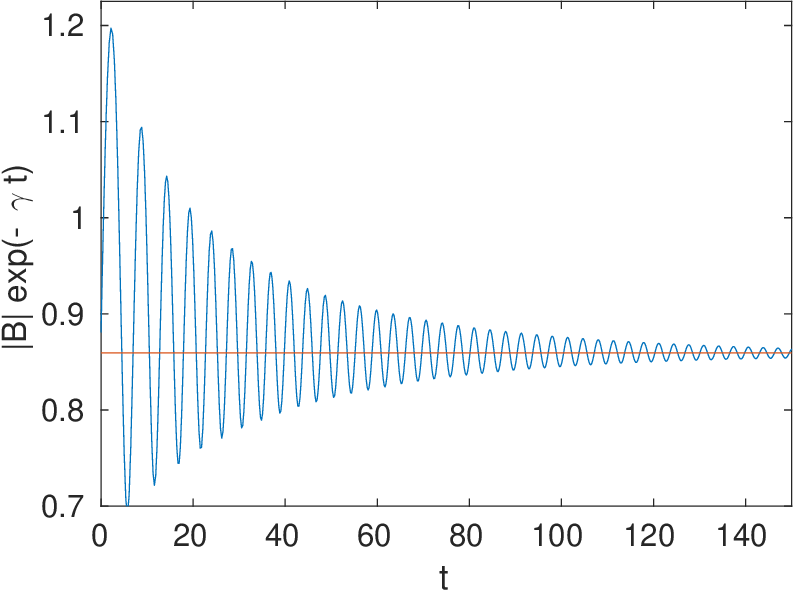}
  }
  \caption{Evolution of $|A|$, $|B|$ and $|B|\e^{-\gamma t}$ for $k=1$, $a=0.45$, $\epsilon=0.05$, $\Gamma=0.01$, $\gamma=0.001$.}
  \label{ABmod451}
  \end{figure}
\begin{figure}[!ht]
  \centerline{
    \includegraphics[width=.32\textwidth]{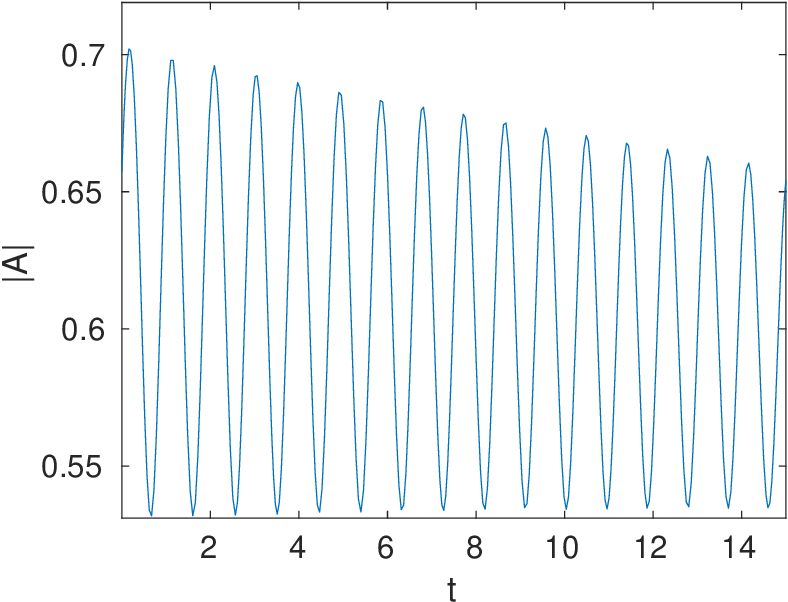}\hspace{5pt}
    \includegraphics[width=.32\textwidth]{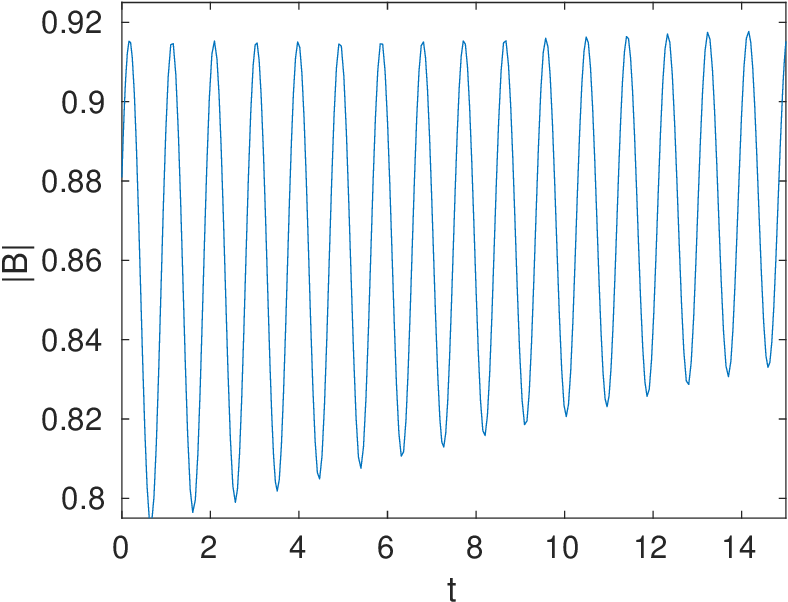}\hspace{5pt}
    \includegraphics[width=.32\textwidth]{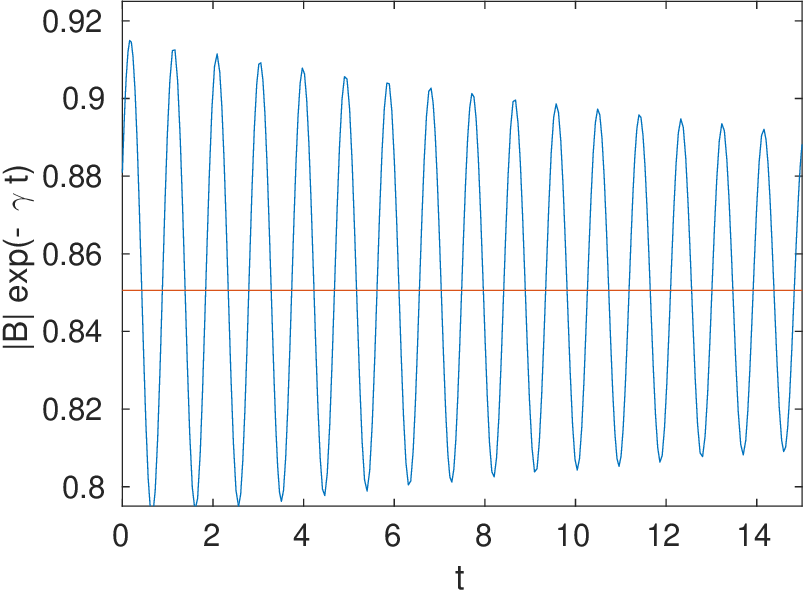}
  }
  \caption{Evolution of $|A|$, $|B|$ and $|B|\e^{-\gamma t}$ for $k=2$, $a=0.55$, $\epsilon=0.05$, $\Gamma=0.01$, $\gamma=0.002$.}
  \label{ABmod552}
  \end{figure}

Figure~\ref{ABmod551}, where $k=1, a=0.55$, shows an initial transient exponential growth, in contrast to Figure~\ref{ABmod552} ($k=2, a=0.55$) and Figure~\ref{ABmod451} ($k=1, a=0.45$) , for which the initial dynamics is oscillatory.  Note that the quantity $4a^2-k^2$ is positive (equal to $0.21$) , when $k=1, a=0.55$, and negative in the other two cases ($-2.79$ and $-0.19$, respectively), suggesting that for sufficiently high amplitude a finite number of modes will exhibit a {\em transient} Benjamin-Feir (BF) type-instability (see beginning of~\S~\ref{transient} for more discussion the BF criterion $4a^2-k^2>0$. Before investigating this regime in more detail, we make some observations about the effects of diffusion on long-time evolution.

\subsubsection*{Diffusion effects on long-time dynamics}
Diffusion can be modeled by adding the higher order term $-\ri \epsilon^2 \delta u_{xx}$, with $\delta>0$, to equation~\eqref{vHONLS}. Linearizing~\eqref{vHONLS} with the added diffusive term changes the matrix $M$ to

\[
M_\delta=\begin{pmatrix} -\gamma-\epsilon^2\delta k^2 & \frac{1}{2}k^2(-2+\epsilon k)  & 0 & 0 \\
-\frac{1}{2}k^2(-2+\epsilon k)   & -\gamma -\epsilon^2\delta k^2 & 0 & 0 \\
0 & 0  & \gamma -\epsilon^2\delta k^2 & -\frac{1}{2}k^2(2+\epsilon k)   \\
0  & 0  & \frac{1}{2}k^2(2+\epsilon k)   & \gamma -\epsilon^2\delta k^2
\end{pmatrix}.
\]

The eigenvalues of $M_\delta$ are 
\[
(-\epsilon^2\delta k^2+ \gamma) \pm \frac{ik^2}{2} (2+\epsilon k), \qquad -(\epsilon^2\delta k^2+\gamma) \pm \frac{ik^2}{2} (2-\epsilon k).
\]
\begin{prop}
  \label{prop2p2}
  Addition of the higher-order diffusive term $-\ri \epsilon^2 \delta u_{xx}$, with $\delta>0$, will stabilize all but a finite number of modes. For $\epsilon \delta$ large enough, all modes will become stable.
\end{prop}
\begin{proof}
  Since $\gamma =2\epsilon k \Gamma$, where $0<\Gamma \ll 1$, the term $(-\epsilon^2\delta k^2+ \gamma)=\epsilon k(-\epsilon \delta k+2\Gamma)$ can be made negative. If we fix $\delta>0$, only a finite number of modes with $|k|\leq \lfloor 2\Gamma/\delta\epsilon \rfloor$ will be unstable, where $\lfloor x \rfloor$ denotes the integer part of $x$ when $x>0$ . In particular, choosing $\delta$ such that $\epsilon \delta > 2\Gamma$ will stabilize all the modes.
\end{proof}

\begin{rem}
The term $-\ri \epsilon^2 \delta u_{xx}$ can be obtained from the model presented in~\cite{didyza2008},  by carrying out the asymptotic derivation of the dissipative terms one order beyond where the Dysthe equation appears. It also appears in the model derived in~\cite{eeltink2019} as a higher-order term leading to high-frequency cut-off when viscosity is present. In fact, Proposition~\ref{prop2p2} characterizes how a small amount of diffusion leads to more physically realistic behavior for large times when viscous effects are present. Figure~\ref{ABmod451D} shows the stabilizing effect of diffusion for the same initial conditions and parameter values as in Figure~\ref{ABmod451}, with $\epsilon \delta= 2.1\Gamma$.
\end{rem}

\begin{figure}[!ht]
  \centerline{
    \includegraphics[width=.32\textwidth]{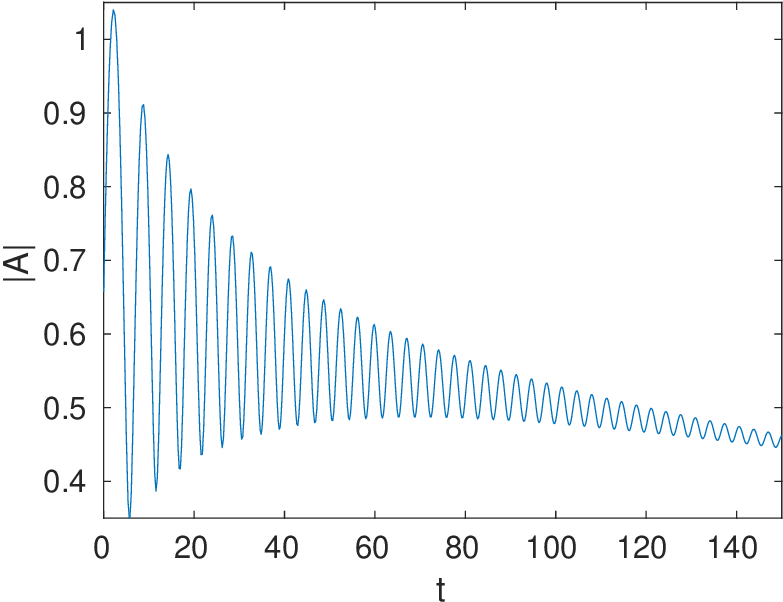}\hspace{30pt}
    \includegraphics[width=.32\textwidth]{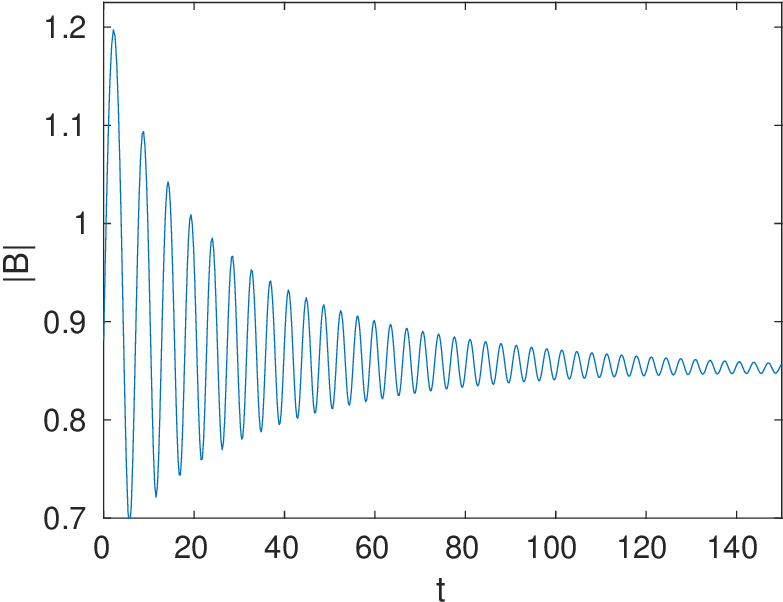} }
  \caption{Evolution of $|A|$ and $|B|$ with diffusion for $k=1$, $a=0.45$, $\epsilon=0.05$, $\Gamma=0.01$, $\gamma=0.001$, $\epsilon \delta=0.021$.}
  \label{ABmod451D}
  \end{figure}

As discussed earlier, multiplying $v(x,t)$ by $e^{-\Gamma t}$ is equivalent to subtracting $\Gamma$ from the real parts of the eigenvalues of $M_\delta$, enhancing stabilization. In particular, we have

\begin{cor}
If $\epsilon \delta>2\Gamma$, then the damped Stokes wave $u^\Gamma_0$ is a linearly stable solution of the viscous HONLS with added diffusion.
\end{cor}

\subsection{Transient Benjamin-Feir Instability}
\label{shortmoder}

For fixed $k\geq 0$, if the amplitude $a$ of the damped Stokes wave $u^\Gamma_a$ is sufficiently large then the amplitudes of both modes, $|A|$ and $|B|$, grow exponentially until a time $t^*$ where an inflection point occurs; at $t^*$, the exponential growth begins to saturate. For $t>t^*$, the modes exhibit oscillatory behavior superimposed with a slow decay or growth induced by the viscous effects. 
Figure~\ref{BShort} (left) shows the evolution of $|B|$ for an example where such transient instability occurs.  This is compared with the purely oscillatory dynamics in Figure~\ref{BShort} (right), corresponding to the solution of the linearization~\eqref{linNLS} about a damped plane wave with amplitude below the threshold for initial instability. We remark that $|A|$ shows a similar behavior for the same choice of parameters.

\begin{figure}[!ht]
  \centerline{
    \includegraphics[width=.3\textwidth]{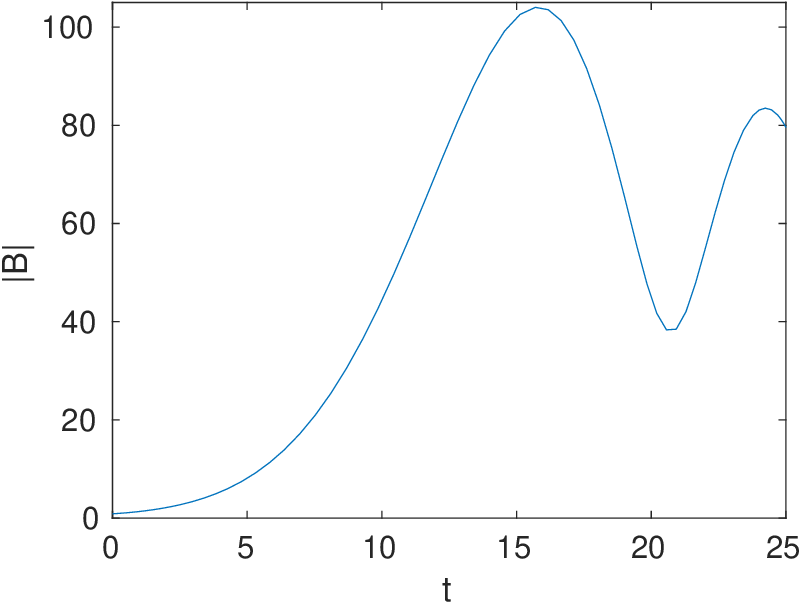}
    \hspace{50pt}
    \includegraphics[width=.3\textwidth]{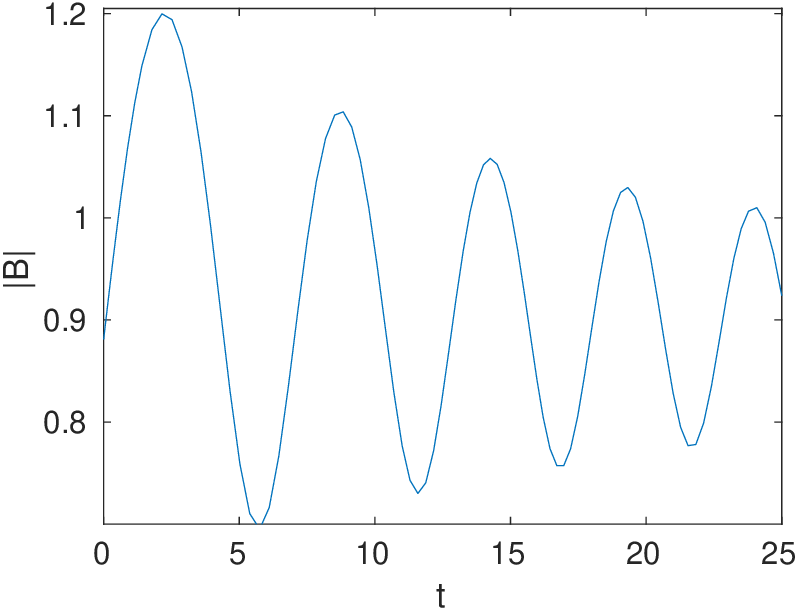} 
  }
  \caption{Initial evolution of $|B|$ for $a=0.55$ (left), exhibiting an initial exponential growth; and for $a=0.45$ (right), exhibiting purely oscillatory dynamics; with $k=1$, $\Gamma=0.01$, $\gamma=0.001$, $\epsilon=0.05.$}
  \label{BShort}
  \end{figure}

Note how selecting an amplitude slightly above or slightly below the apparent threshold $a\sim 0.5$ leads to substantially different initial growths, with $|B|$ increasing to its first maximum value of about $100$  when $a=0.55$ and of about $1.2$ when $a=0.45$.  Also, in the presence of viscosity, the first maximum is the absolute maximum of $|B|$ over a very long time interval.  

We also remark that multiplying by $\e^{-\Gamma t}$ does not change qualitatively the initial time behavior and the transition to oscillatory dynamics, as shown in Figure~\ref{|B|EG} for the same parameter values as in Figure~\ref{BShort} (left). In fact, the first maxima of the initially unstable modes are also their global maxima and will determine the largest deviation of perturbed solution from $u^\Gamma_a$. Thus, the transient dynamics due to the initial Benjamin-Feir instability is an important component of the linear stability analysis of the damped Stokes wave. 

\begin{figure}[!ht]
   \centerline{
    \includegraphics[width=.34\textwidth]{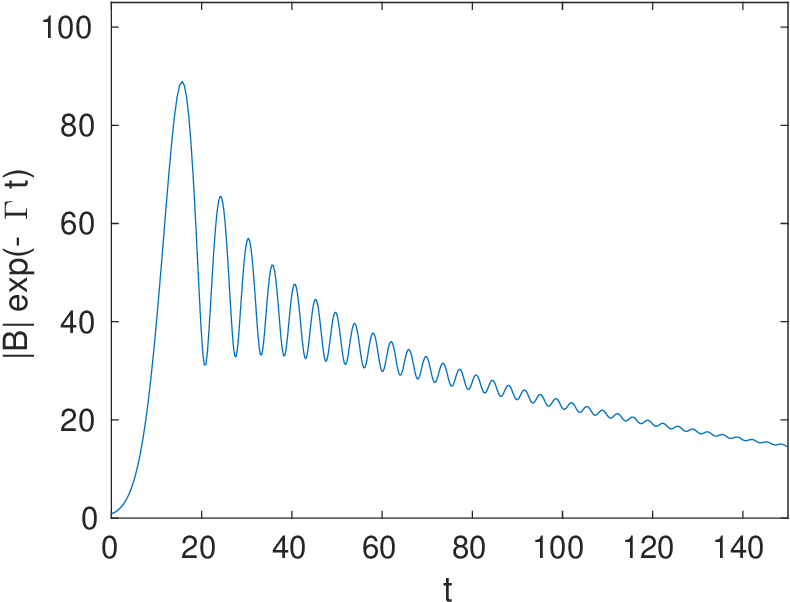}
    }
   \caption{Evolution of $|B|\e^{-\Gamma t}$ for $k=1$, $a=0.55$, $\epsilon=0.05$, $\Gamma=0.01$, $\gamma=0.001$.}
   \label{|B|EG}
   \end{figure}

\subsubsection{Qualitative Analysis}
\label{transient}

We next examine the transition from exponential growth to oscillatory dynamics for solutions of the linearization of the Dysthe equation~\eqref{linNLS} with and without viscosity. (The discussion with added diffusion is similar and will be left to the reader.) In what follows, we will assume that the initial amplitude of the damped Stokes wave solution~\ref{dStokes}, satisfies an inequality analogous to the Benjamin-Fair instability criterion $4a^2-k^2>0, k=1, \ldots M, M\in \mathbb{N}$ for the original NLS model ($\epsilon=0$.). This guarantees that $u_a^\Gamma$ is initially unstable with $M$ unstable modes. The case $M=1$ is chosen to illustrate the initial exponential growth in the figures in this section.

To study the transient instabiity,  it is convenient to substitute the complex expansion $v(x,t)=A(t)\e^{\ri kx}+B(t) \e^{-\ri kx}$ of the $k$-th Fourier mode of $v$ (where the subscript $k$ has been dropped) in~\eqref{linNLS} to obtain the following system
\begin{equation}
\label{ABeq}
\begin{split}
\ri A_t + [a(t)+b(t)]A + c(t) B^* & =0, \\
\ri B_t + [a(t)- b(t)]B + c(t) A^* & =0,
\end{split}
\end{equation}
where 
\[
a(t)=-k^2 +c(t), \qquad b(t)=\frac{\epsilon}{2} k^3 + 8 \epsilon k |u_a^\Gamma|^2 + \ri \gamma, \qquad c(t)=2|u_a^\Gamma|^2 (1 + \epsilon k),
\]
and where the superscript ${}^*$ denotes complex conjugation.

Rescaling $A= \displaystyle e^{\ri \int^t b(s) \rd s}C, B= \displaystyle e^{-\ri \int^t b(s) \rd s}D$, we get
\begin{equation}
\label{CDeq}
\begin{split}
\ri C_t + a(t)C + c(t) \e^{-\ri \int^t (b-b^*)(s) \rd s}D^* & =0, \\
\ri D_t + a(t)D + c(t) \e^{\ri \int^t (b-b^*)(s) \rd s}C^* & =0,
\end{split}
\end{equation}
with $b(t)-b^*(t)=2\ri\gamma.$

\subsubsection*{Non-viscous case}
When $\gamma=0$, system~\eqref{CDeq} reduces to 
\[
\begin{split}
\ri C_t + a(t)C + c(t) D^* & =0, \\
\ri D_t + a(t)D + c(t) C^* & =0.
\end{split}
\]
Setting $U=C+D^*$, $\ri V= C-D^*$, we arrive at the system
\[
\begin{split}
U_t -k^2 V & =0, \\
V_t -[-k^2+4a^2(1+\epsilon k) \e^{-2\Gamma t}]U & =0,
\end{split}
\]
and the equivalent second order differential equation for $U$
\begin{equation}
\label{2de0}
    U_{tt}+k^2[k^2-4a^2(1+\epsilon k)\e^{-2\Gamma t}]U=0.
\end{equation}
While $U$ is complex, equation~\eqref{2de0} has real coefficients, thus one can apply Sturmian theory as in~\cite{segurstab} to conclude that:
\begin{prop} 
\label{0transition}
When $\gamma=0$
\begin{itemize}
\item[-] If $\displaystyle 4a^2<\frac{k^2}{1+\epsilon k}$, then all solutions of~\eqref{2de0} are oscillatory.
\item[-]
If $\displaystyle 4a^2-\frac{k^2}{1+\epsilon k}>c^2$, for some non-zero constant $c$, then there exists a solution of~\eqref{2de0} which grows exponentially for $\displaystyle 0<t<\frac{1}{2\Gamma}\ln  \frac{4a^2(1+\epsilon k)}{k^2}$, and exhibits bounded oscillations  for $\displaystyle t\in \Big(\frac{1}{2\Gamma}\ln \frac{4a^2(1+\epsilon k)}{k^2}, +\infty\Big)$.
\end{itemize}
\end{prop}
\begin{rem}
Since $|A|=|C|$ and $|B|=|D|$, the same conclusions hold for solutions $A$ and $B$ (the coefficients of the $k$-th Fourier mode) of equation~\eqref{ABeq}. Thus, compared with the linearly damped NLS case studied in~\cite{segurstab}, the higher order terms in~\eqref{vHONLS} {\em increase the instability} of an initially unstable damped Stokes wave (by effectively increasing its amplitude by the factor $\sqrt{1+\epsilon k}>1$). Furthermore,  the time of onset of the oscillatory behavior is delayed. (Compare $\displaystyle t^*_0=\frac{1}{2\Gamma}\ln \frac{4a^2(1+\epsilon k)}{k^2}$ with its expression at $\epsilon=0$.)
\end{rem}

\subsubsection*{Viscous case}
When $\gamma \not=0$, system~\eqref{CDeq} becomes
\[
\begin{split}
\ri C_t + a(t)C + c(t) \e^{2\gamma t} D^* & =0, \\
\ri D_t + a(t)D + c(t) \e^{-2\gamma t}  C^* & =0.
\end{split}
\]
With the substitutions $U=C \e^{-\gamma t}+D^* \e^{\gamma t}$, $\ri V=C \e^{-\gamma t}-D^* \e^{\gamma t}$, we obtain
\[
\begin{split}
U_t +(-k^2 + \ri \gamma) V & =0, \\
V_t -[-k^2++ \ri \gamma + 4a^2(1+\epsilon k) \e^{-2\Gamma t}]U & =0,
\end{split}
\]
or equivalently
\begin{equation}
\label{2deg}
    U_{tt}+(k^2 - \ri \gamma) [k^2 -\ri \gamma-4a^2(1+\epsilon k)\e^{-2\Gamma t}]U=0.
\end{equation}
Note that, since $|A|=\e^{-\gamma t} |C|$ and $|B|=\e^{\gamma t} |D|$, it follows from Section~\ref{largetime} that the long-time dynamics of $|C|$ and $|D|$ will be oscillatory.  

Making the change of variable $z=-\Gamma t$ and letting $\lambda^2=-4(k^2-\ri \gamma) (1+\epsilon k) a^2/\Gamma^2$, $\nu^2=-(k^2-\ri \gamma)^2/\Gamma^2$, equation~\eqref{2deg} is cast in one of the standard forms of Bessel differential equations (see {\bf 10.13.6} in~\cite{OLBC}):
\[
U_{zz}+(\lambda^2 \e^{2z}-\nu^2)U=0,
\]
whose solutions can be written as $\mathcal{C}_\nu(\lambda \e^z)$, where $\mathcal{C}_\nu$ is any of the cylinder Bessel functions of order $\nu$---$Y_\nu$, $J_v$ and the Hankel functions $H_\nu^{(1)}, H_\nu^{(2)}$---  or linear combinations thereof.

By using the limiting forms of the cylinder Bessel functions
$\displaystyle \mathcal{C}_\nu (z)  \sim (\tfrac{1}{2} z)^{\pm\nu}$ for small argument $z$ and fixed order $\nu$ with $\Re(\nu)>0$ (see {\bf 9.1.7–9.1.9} in ~\cite{AS}), we obtain the long-time asymptotic expression
\[
U\sim \left(\tfrac{1}{2} \lambda \e^{-\Gamma t}\right)^{\pm(\gamma -\ri k^2)/\Gamma}\sim \e^{\mp (\gamma -\ri k^2)t},
\]
describing oscillatory behavior superimposed upon slow exponential growth or decay due to the viscous term. 
\smallskip

To our knowledge, there is no general ``Sturmian" theory for 
second order equations with complex variable coefficients. See  E. Hille's book~\cite{hille1976ordinary} for comparison theorems and non-oscillation theory for such equations in the complex domain, and Barrett's survey~\cite{barrett} for the case of real independent variable. However, for equation~\eqref{2deg}, we can prove the following
\begin{thm}
\label{gtransition}
Let $\gamma\not=0$ and let $\displaystyle t^*_0=\frac{1}{2\Gamma}\ln \frac{4a^2(1+\epsilon k)}{k^2}$ be the transition time for the zero-viscosity case from Proposition~\ref{0transition}.
\begin{itemize}
\item[-] If $\displaystyle 4a^2<\frac{k^2}{1+\epsilon k}$ , then all solutions of~\eqref{2deg} are oscillatory.
\item[-]
If $\displaystyle 4a^2-\frac{k^2}{1+\epsilon k}>c^2$, $c\not=0$, then for sufficiently small  $\gamma>0$ there  exists a solution of~\eqref{2deg} which: 1. grows exponentially for $0<t<t^*_\gamma$; and  2.  exhibits a  predominantly oscillatory dynamics, superimposed upon exponential growth or decay at rate $\gamma$, for $t>t^*_\gamma.$ 

\item[-] The transition time $t^*_\gamma$ (defined as the location of the first inflection point of $|U|$), satisfies
 $t^*_0<t^*_\gamma \leq t^*_0 + C \gamma^2+O(\gamma^4)$, where  $\displaystyle C=\frac{4[2a^2(1+\epsilon k)-k^2]^2 t^*_0+1}{\Gamma k^2}$. 
\end{itemize}
\end{thm}
\begin{rem} It follows from Theorem~\ref{gtransition}, that a small amount of viscosity does not change the qualitative behavior of a solution of the linearization, but extends the interval of instability of an initially unstable Stokes waves, thus
delaying the onset of oscillations. 
\end{rem}

\begin{proof}
Throughout the proof, we assume that $\displaystyle 4a^2-\frac{k^2}{1+\epsilon k}>0$. In order to capture the time at which the complex modulus of $U$ has an inflection point, we introduce polar coordinates $U=\rho \e^{\ri \phi}$ and study the resulting coupled system of real ODEs:
\begin{equation}
\label{polar}
\begin{split}
\rho_{tt} -\rho [f(t)+\theta_t^2+\gamma^2] &= 0, \\
\rho \theta_{tt}+ 2 \rho_t \theta_t + g(t) \rho & =0,
\end{split}
\end{equation}
where $f(t)=k^2(2c(t)-k^2)$ and $g(t)=2\gamma (c(t)-k^2)$, with $c(t)=2 a^2 (1 + \epsilon k) \e^{-2\Gamma t}$.
\smallskip

Set $G(t;\gamma)=f(t)+\theta_t^2+\gamma^2$. First observe that, when $\gamma=0$ system~\eqref{polar} reduces to
\begin{equation}
\label{polar0}
\begin{split}
\rho_{tt} -\rho [f(t)+\theta_t^2] &= 0, \\
\rho \theta_{tt}+ 2 \rho_t \theta_t & =0.
\end{split}
\end{equation}
Choosing the initial conditions $\rho(0)=\rho_0, \, \rho_t(0)=\rho_1,\, \theta(0)=\theta_t(0)=0$, corresponding to a real $U(0)$, system~\eqref{polar0} reduces to 
\[
\rho_{tt}-k^2[4a^2(1+\epsilon k)e^{-2\Gamma t} -k^2] \rho=0,
\]
the same as equation~\eqref{2de0}. 

It follows from Proposition~\ref{0transition} for the non-viscous case that, if  $\displaystyle 4a^2>\frac{k^2}{1+\epsilon k}$ there exists a solution of~\eqref{polar0} which grows exponentially until time $t^*_0$ where $G(t^*_0; 0)=0$, and then becomes oscillatory for $t>t^*_0$; here $\displaystyle t^*_0=\frac{1}{2\Gamma}\ln \frac{4a^2(1+\epsilon k)}{k^2}=\frac{1}{\Gamma}\ln  \frac{2a\sqrt{1+\epsilon k}}{|k|}$.
\smallskip

When $\gamma\not=0$, consider system~\eqref{polar} with real initial conditions, i.e., $\theta(0)=\theta_t(0)=0$. 
Since $G(t;\gamma)=G(t;0)+ (\theta_t)^2+\gamma^2$, we have that $G(t;\gamma)>\gamma^2$ on the interval $(0, t^*_0)$, so Sturmian Theory guarantees the existence of a solution $(\rho,\theta)$ such that $\rho$ grows exponentially at least as fast as $\e^{\gamma t}$. More precisely, 
\begin{equation}
\label{estg}
\rho(s+t)\geq \rho(t) \e^{\gamma s}, \quad \forall \ s,t>0, \ {\mathrm{with}} \  t+s\leq t^*_0.
\end{equation}

From the smoothness of the vector field of~\eqref{2deg} we have that $G(t;\gamma)$ is continuously differentiable in a neighborhood of $(t^*_0, 0)$. Moreover,  $G_t(t^*_0;0)=-2\Gamma k^2\not=0$. It then follows from the Implicit Function Theorem that for sufficiently small $\gamma$ there exists a unique $t^*_\gamma$ in a neighborhood of $t^*_0$ such that $G(t^*_\gamma;\gamma)=0$. Since $G_t(t^*_0;0)<0$, then we must have $t^*_\gamma>t^*_0$: thus the effect of viscosity is to increase the initial interval of instability and delay the onset of oscillations.

We can estimate $t^*_\gamma$ using the linearization 
\begin{equation}
\label{linG}
\begin{split}
G(t;\gamma) & = G(t^*_0;\gamma) + G_t(t^*_0;\gamma) (t - t^*_0) + O((t-t^*_0)^2) \\
& =\theta_t^2(t^*_0)+\gamma^2 +\left[-2 \Gamma k^2 c(t^*_0) - \big( 2\frac{\rho_t}{\rho} \theta_t^2 + g\theta_t\big) \right]_{t=t^*_0} (t-t^*_0) + O((t-t^*_0)^2).
\end{split}
\end{equation}

An upper bound for $\theta_t^2(t^*_0)$ can be found by using the second equation in system~\eqref{polar} with $\theta_t(0)=0$, which gives
\begin{equation}
\label{thetadot}
\theta_t=\frac{\int_0^t g \rho^2 \rd s}{\rho^2}.
\end{equation}
Using~\eqref{estg} with $s+t=t^*_0$, we get
\[
\begin{split}
|\theta_t(t^*_0)|\leq & \max_{(0,t^*_0)}|g(t)|\frac{\int_0^{t^*_0} \rho^2 \rd s}{\rho^2(t^*_0)} \leq |g(0)| \frac{\int_0^{t^*_0} \rho^2(t^*_0) \e^{-2\gamma (t^*_0-t)} \rd t}{\rho^2(t^*_0)}=\frac{|g(0)|}{2\gamma}\left( 1 - \e^{-2\gamma t^*_0}\right) \\
&=2\gamma |2a^2(1+\epsilon k)-k^2| t^*_0 + O(\gamma^2).
\end{split}
\]
Since $|\rho_t|/\rho$ is bounded and $g\theta_t=O(\gamma^2)$, by truncating~\eqref{linG} at first order, we get
\[
t^*_\gamma=t^*_0+\frac{\theta_t^2(t^*_0)+\gamma^2}{\Gamma k^4} +O(\gamma^4).
\]
It follows that the time $t^*_\gamma$ at which $\rho$ transition to a predominantly oscillatory behavior satisfies
\[
t^*_0<t^*_\gamma \leq t^*_0+\gamma^2 \frac{4[2a^2(1+\epsilon k)-k^2]^2 t^*_0+1}{\Gamma k^2} + O(\gamma^4).
\]
\end{proof}
\begin{rem}
For $k=1, a=.55, \epsilon=.05, \Gamma=.01, \gamma=.001$, $t^*_0=11.9705$ and the upper estimate for the transition time is $t^*_\gamma=11.9712$. The numerically computed inflection time is $11.97075$, in agreement with the estimate provided by the linear approximation. 
\end{rem}

\section{Frequency downshifting and rogue wave activity} 
\label{FDS_RWA}

A main feature of the HONLS model, equation~\eqref{vHONLS} with $\Gamma = 0$, $\delta = 0$,
is the nonlinear interaction and subsequent energy transfer between the various modes 
comprising its solutions.
When considering perturbed Stokes wave initial data, solutions of the HONLS equation
 exhibit
temporary frequency downshifting  with the wavetrain undergoing successive modulations and demodulations as the energy is transferred back and forth between the carrier and sideband modes.
 \smallskip
 
 In this section we examine the impact of viscous damping on frequency
 downshifting in water waves.  Specifically, we investigate the mechanism behind permanent downshift, estimate the time of permanent downshift, and explore the relation between frequency downshifting and rogue wave formation.
 The numerical studies are carried out for the following models
 obtained from equation \eqref{vHONLS} and denoted as:

\begin{equation}
\begin{array}{lcl}
LDHONLS && \mbox{linear damped HONLS:  $\epsilon = 0.05$, $\Gamma \ne 0$, $\nu = 0$, $\delta = 0$;}\\ \\
vHONLS && \mbox{viscous HONLS without diffusion:  $\epsilon = 0.05$, $\Gamma \ne 0$, $\nu = 1$, $\delta = 0$;}\\ \\
dvHONLS && \mbox{viscous HONLS with diffusion:  $\epsilon = 0.05$, $\Gamma \ne 0$, $\nu = 1$, $\delta = 42\Gamma$
  (i.e., $\epsilon\delta  = 2.1\Gamma$).}
\end{array}
\label{block}
\end{equation}

\subsection{Measuring Downshifting}
\label{measure}

We begin by discussing 
how frequency downshifting is measured. We then provide a criterion for permanent downshifting and apply it to the various models.
In addition,  since we are interested in understanding
how rogue wave activity is affected by viscosity and downshifting,
we recall a standard criterion for rogue waves.
\smallskip

The wave energy $E$ and the momentum, or energy flux, $P$ are defined as follows
\begin{equation}
    E(t) = \frac{1}{L} \int_0^L |u|^2 dx,\qquad \qquad    P(t) = \frac{i}{2L} \int_0^L (u u^*_x - u_x u^*) dx,
\end{equation}
where $L$ is the spatial period of the solution $u(x,t)$.
Using the Fourier series expansion
\begin{align*}
u(x,t) = \sum_{k = -\infty}^\infty \hat{u}_k(t) e^{ikx},
\end{align*}
we write the Fourier representations
\begin{align}
 E(t)  = \sum_{k = -\infty}^\infty |\hat{u}_k|^2, \qquad  P(t)  = \sum_{k = 1}^\infty k(|\hat{u}_k|^2 -|\hat{u}_{-k}|^2).
\end{align}
 Significantly, the momentum measures the asymmetry of the Fourier modes, which
is a key element in understanding frequency downshifting.
\smallskip

Two different spectral quantities, the  {\em spectral center} and the {\em spectral peak}, are commonly used to measure frequency downshifting. 

\begin{dfn} {\ }

\begin{itemize}
    \item[-]

The spectral center
\begin{equation}
    k_m (t): =  \frac{\sum_{k = -\infty}^\infty k|\hat{u}_k|^2}{\sum_{k = -\infty}^\infty |\hat{u}_k|^2} = \frac{P(t)}{E(t)},
\end{equation} 
also referred to as the  spectral mean, 
is a weighted average of the spectral content of the wave \cite{uchiyama}.

\item[-] The spectral peak $k_{peak}(t)$ is the wave number of the Fourier
  mode of maximal amplitude at time $t$. In other words, $k_{peak}$ corresponds to the  highest energy mode.  
\end{itemize}
\label{kmdef}
\end{dfn}

In terms  of the spectral mean, frequency downshifting is considered to occur when
$k_m$ is monotonically decreasing, while in the spectral peak sense downshifting  occurs  when
the spectral peak changes from its original value $k_{peak} = 0$ to a lower mode \cite{uchiyama, trulsen97}.

The spectral diagnostics related to $k_m$ and $k_{peak}$ 
 measure different characteristics of the wave energy distribution.
When the Fourier components are sufficiently concentrated about 
$k_{peak}$, one has $k_m \approx k_{peak}$. 
However, the same value of $k_m$
can result from different distributions of $|\hat{u}_k|^2$, leading to differing claims of downshifting behavior.

\subsubsection*{A criterion for frequency downshifting.}
 \label{def_DS} 

To obtain a more complete description of the spectral activity
   we monitor both $k_m$ and $k_{peak}$ in the  numerical experiments  and 
adopt the following criterion:

{\em Frequency downshifting} occurs if both of the following conditions are met:
  \begin{enumerate}
  \item $\displaystyle \frac{d k_m}{dt} < 0$, and
  \item $k_{peak}$ moves permanently to a lower mode.
    \end{enumerate}

  If both conditions are met, our focus turns to an aspect of downshifting which is related to $k_{peak}$.
  Since $k_{peak}$ may shift to a sideband mode and then
  back to the carrier mode
 repeatedly while $k_m$ is monotonically decreasing, we are interested in
 identifying the last  time at which the carrier mode is the dominant mode.

\begin{dfn}
  \label{tds}
  Provided the two conditions for frequency downshifting are satisfied, 
the {\em time of permanent downshift} $t_{DS}$ is said to be 
the last time at which the carrier wave is the dominant mode, i.e.
$t_{DS}$  is the last time $k_{peak} = 0$ .
\end{dfn}

\begin{rem}
For the range of $\Gamma$-values and initial data considered
in the numerical experiments, the conditions on  $k_m$ and $k_{peak}$
for frequency downshifting are typically both met.
However, there
are   ``outlier'' experiments where only one condition is met and which
can be classified as one  of the
following scenarios:
1) In LDHONLS experiments one can have $\displaystyle \frac{d k_m}{dt} = 0$ while $k_{peak}$ either upshifts or downshifts, as in the example shown in  Figure~\ref{LD_out}.
2) In vHONLS experiments, when the damping parameter $\Gamma$ is very small, we can have $\displaystyle \frac{d k_m}{dt} < 0$, while $k_{peak}$ does not permanently move to a lower mode, even over very long simulation time frames. Based on the criterion adopted here, the "outlier" experiments are not considered
to  exhibit frequency downshifting.
\end{rem}

For the vHONLS and dvHONLS equations one finds 
\begin{align}
  \frac{dE}{dt} = -2\Gamma E(t) - 4\nu\epsilon\Gamma P(t) - 2\epsilon^2\delta Q(t),
\label{derivE}
\end{align}
\begin{align}
\frac{d P}{dt} = -2\Gamma P(t) - 4\nu\epsilon\Gamma Q(t) -2\epsilon^2\delta R(t),
\label{derivP}
\end{align}
and
\begin{align}
  \frac{dk_m}{dt} = \frac{-4\nu\epsilon\Gamma}{E^2}\left(EQ - P^2\right) + \frac{2\epsilon^2\delta}{E^2}\left(PQ - ER\right),
\label{derivkm}
\end{align}
where
\[
Q(t) = \frac{1}{L}\int_0^L |u_x|^2 dx, \qquad \qquad R(t) = \frac{\ri}{2L}\int_0^L (u_x u_{xx}^* - u_x^* u_{xx})\,dx.
\]

Given that $E$ and $P$ are conserved quantities for both the  NLS  ($\epsilon = \Gamma = 0$)
and the 
HONLS  ($\Gamma = \delta = 0$) equations, the spectral center $k_m$ is also
independent of time for these models.
Moreover,  $k_m$ is constant
for the LDHONLS ($\nu = 0$, $\delta = 0$) equation, since
$E$ and $P$  decrease  at the same exponential rate.
\smallskip

We note that when $\delta = 0$, equations~\eqref{derivE} - \eqref{derivP}  do not reduce
to the equations obtained  in \cite{cartergovan} for the evolution of
$E$ and $P$. The difference arises due to the following  factors: i) 
A different coordinate transformation is employed in \cite{cartergovan} than the one used to derive equation (1.1), (see the Appendix).
ii) The vHONLS describes the evolution of the leading term
in the expansion for the velocity potential and is  a version of the dissipative Gramstad-Trulsen equation.   When $\Gamma = \delta = 0$ equation~(\ref{vHONLS}) is Hamiltonian.
On the other hand, \cite{cartergovan} focuses on  the system representing
the evolution of the leading term  in the expansion for surface displacement 
which, without the dissipative and diffusive terms,  is likely not Hamiltonian.


When viscosity  is  present ($\nu = 1$), with or without diffusion, $E$ and $P$  decrease at different rates allowing for the possibility of downshifting. From equation \eqref{derivkm}, downshifting in the {\em spectral mean sense}  will occur under the following conditions:
\begin{equation}
    \begin{array}{rl}
      (i) & \mbox{In the vHONLS case ($\delta =0$),  since $EQ - P^2$ is positive by the  Cauchy Schwartz inequality,}\\ \\
          & \mbox{then  $\displaystyle\frac{dk_m}{dt} < 0$, $\forall t$ and $\forall \Gamma$.}\\ \\
(ii) & \mbox{In the dvHONLS case (with $\delta =\alpha\Gamma$, for some $\alpha >0$):}\\ \\
& -\mbox{If $PQ + ER < 0 $ then   $\displaystyle\frac{dk_m}{dt}$ < 0, $\forall t$ and $\forall \Gamma$.}\\ \\
& -\mbox{If $PQ + ER > 0$ then $\displaystyle\frac{dk_m}{dt} < 0$ provided
        the $\mathcal{O} (\epsilon^2\Gamma)$ diffusive correction is smaller than}\\ \\
      & \mbox{\hspace{8pt} the $\mathcal{O} (\epsilon\Gamma)$ term.}
\end{array}
    \label{key}
\end{equation}

The contribution of the individual modes to the evolution of $P$ and $E$ can be seen by turning to their Fourier representations:
\begin{align}
\label{FdEdt}
 \frac{dE}{dt} = -2\Gamma\left[ |\hat{u}_0|^2 + \sum_{k = 1}^\infty (|\hat{u}_k|^2(1 + 2\nu\epsilon k)  + |\hat{u}_{-k}|^2(1 - 2\nu\epsilon k)) \right] -2\epsilon^2\delta \sum_{k=-\infty}^\infty k^2|\hat{u}_k|^2,
\end{align}

\begin{align}
\label{FdPdt}
\frac{dP}{dt} = -2\Gamma \sum_{k = 1}^\infty k\left[|\hat{u}_k|^2(1 + 2\nu\epsilon k) -|\hat{u}_{-k}|^2(1 - 2\nu\epsilon k)\right] + 2\epsilon^2\delta \sum_{k=1}^\infty k^3\left(|\hat{u}_k|^2-|\hat{u}_{-k}|^2\right).
\end{align}
Significantly,  for the vHONLS model,  equation~(\ref{key}) shows that  the first condition for frequency downshifting,  $\displaystyle \frac{d k_m}{dt} < 0$, is
satisfied  $\forall  \Gamma$.
In Section~\ref{visc_results}, we study the structure  of the linear momentum $P$ in order to determine the underlying mechanism responsible for  permanent downshift in the vHONLS model.
Using
equation (\ref{FdPdt}) we provide an estimate for
the location of the global minimum of $P$, which we
find plays a crucial role in determining the time $t_{DS}$ of permanent downshift.

\subsubsection*{A criterion for rogue waves.}

In the numerical experiments we  adopt the following widely used criterion for the occurrence of rogue waves:

Let 
\begin{equation}
    S(t) = \frac{U_{max}(t)}{H_s(t)}
\end{equation}
be the wave strength, 
where $U_{max}(t) = max_{x\in [0,L]} |u(x,t)|$ and $H_s(t)$ is the significant
wave height, defined as four times the standard deviation of the surface elevation. One says that a rogue wave occurs at time $t^*$, if $S(t^*) \ge 2.2$. 

The vHONLS ~\eqref{vHONLS} is derived form the equations governing deep water surface waves where $\eta$ is the 
  free surface displacement (see the Appendix). At leading order the relationship between $\eta$ and $u$ is given by
$\eta = \frac{i\eps u}{2\sqrt{2}k} + (*)$.
As a result, rogue waves observed in the velocity potential correspond to rogue waves in the surface elevation.

\subsection{Frequency downshifting and rogue wave results}
\label{DSresults}
Equation (\ref{vHONLS}) is solved numerically using an exponential time-differencing fourth-order Runge-Kutta scheme, which uses Pad\'e approximations of the matrix exponential terms \cite{kmwy09,lkx15}.  The number of Fourier modes and the time step used in depend on the specific solution under study.
The HONLS equation ($\epsilon =0.05$, $\Gamma =0$, $\delta =0$) is used as a benchmark to test the accuracy of the integrator. As an example, $N = 256$ Fourier modes and a time step $\Delta t = 10^{-3}$ are used for initial data
(\ref{IC-A}). 
With this resolution, Figures~\ref{M_P} (a) -- (c) show that the invariants $E$, $P$ and
${\cal H}$ 
are  conserved with an accuracy  of  $\mathcal{ O} (10^{-13})$,
$\mathcal{ O} (10^{-12})$, and $\mathcal{ O} (10^{-12})$, respectively.
Here the  Hamiltonian for the HONLS equation is given by
\begin{equation}
{\cal H}[u] =  \int_0^L \left\{ | u_x |^2 - | u |^4 - \ri\frac{\epsilon}{4}
\left( u_x u_{xx}^* - u_x^* u_{xx} \right) 
+ 2\ri \epsilon | u |^2\left(u^* u_x - u u_x^*\right) - \epsilon |u|^2
\Big[ H \big( | u |^2 \big) \Big]_x\right\}\,\rd x.
\end{equation}

\begin{figure}[!ht]
  \centerline{
    \includegraphics[width=.3\textwidth]{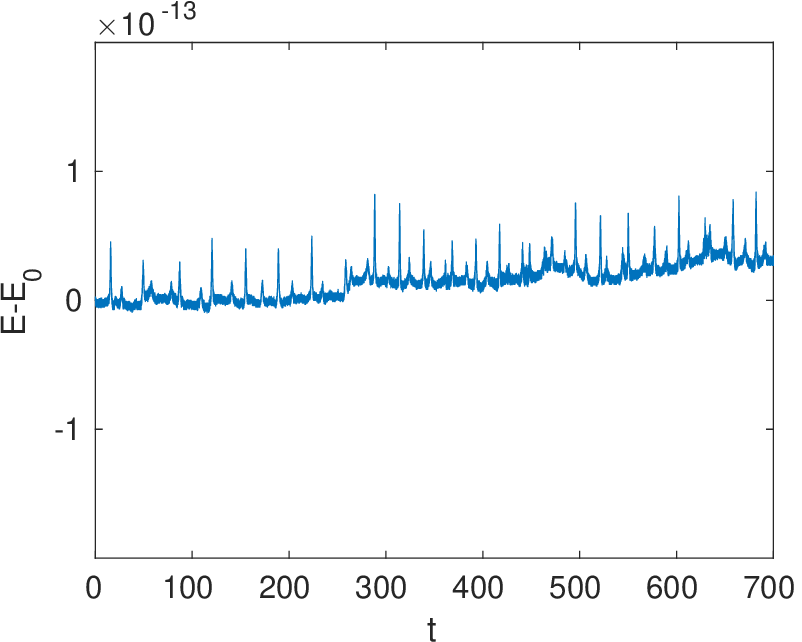}\hspace{12pt}
    \includegraphics[width=.3\textwidth]{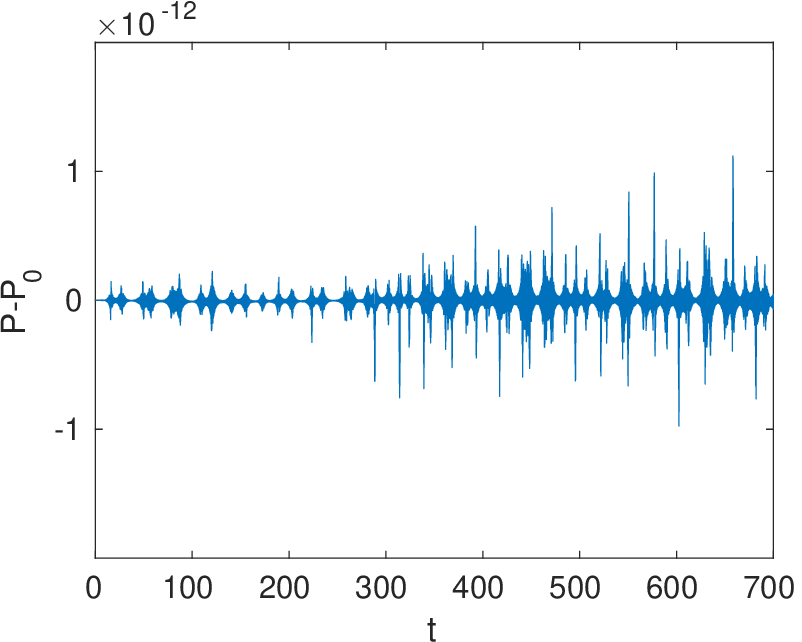}\hspace{12pt}
    \includegraphics[width=.3\textwidth]{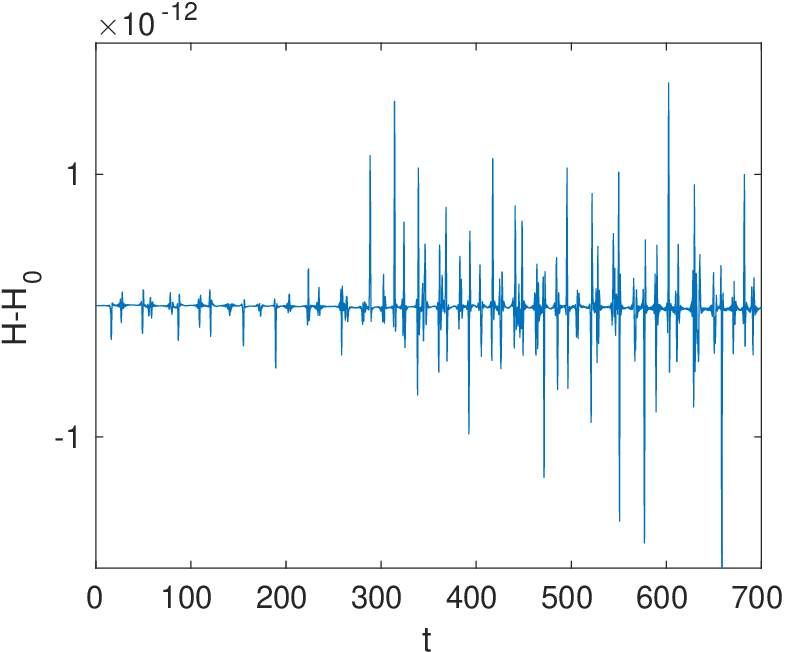}
  }
  \caption{For the HONLS equation with $\epsilon = 0.05$ and initial data (\ref{IC-A}), the figures show preservation of invariants: (a) $E(t)-E(0)$, (b) $P(t)-P(0)$, and (c)
    ${\cal H}(t)-{\cal H}$(0).}
  \label{M_P}
  \end{figure}

\subsubsection*{Choice of initial data and values of $\Gamma$.} Laboratory and numerical experiments have shown that modulated periodic wavetrains evolve chaotically in the presence of two or more unstable modes
\cite{ahhs01}. A chaotic background  can significantly enhance  the likelihood of
rogue waves \cite{cs02}. 
The following family of perturbations of an 
unstable Stokes wave:
\[
  u(x,0) = a(1+\alpha \cos{\mu x}),
\qquad 10^{-2} \le \alpha \le 10^{-1}, \, \mu = \frac{2\pi}{L},
\]
is an ideal choice of initial data for investigating both frequency downshifting and rogue wave formation.
The values of the amplitude $a$ and spatial period $L$ are chosen so that 
the Stokes wave has two unstable modes at $t=0$, i.e.
$0.42 < a < 0.5$ and $L = 4\sqrt{2} \pi$.

For each model, an ensemble of experiments is performed by  varying the damping parameter $\Gamma$ within the range of small-to-moderate values $2.5 \times 10^{-4} \le \Gamma \le 10^{-2}$, with increment $\Delta\Gamma = 2.5\times 10^{-4}$.
We restrict our  choice of $\Gamma$ to this range for the following reasons: 1) One of our goals is to study the effect of viscous damping on rogue wave formation. We do not consider $\Gamma > 10^{-2}$ as  rogue waves are not observed in the evolutions of vHONLS and dvHONLS for 
$0.008 \le \Gamma \le 0.01$, see Figures~ (\ref{NRWs})-(\ref{time_tlrw}). 2) For this range of $\Gamma$ the viscosity-related instabilities do not have a significant effect on the solution.

\begin{rem}
The viscous instabilities have a small growth rate $\gamma = 2\epsilon\Gamma k$.  The important features of the vHONLS and dvHONLS evolutions, such as the time of permanent downshift, occur
  much earlier than the time $\gamma^{-1}$ at which the  viscous instabilities have a noticeable effect. As a consequence, the results for the vHONLS and dvHONLS numerical evolutions are  close.
    For higher values of $\Gamma$, there can be significant differences in the behavior of vHONLS and dvHONLS. Further downshifting may not be observed.
 For example, for the vHONLS with initial data (\ref{IC-A}), $\epsilon = 0.05$, $\Gamma = 0.1 $, one finds  $k_{peak}= 0$  until the numerical solution deteriorates with the high-frequency modes now dominant. Introducing  diffusion in the experiments with $\epsilon \delta = 2.1 \Gamma = 0.21 $ controls the growth of the viscous instabilities, but
 downshifting is not observed.
\end{rem}
 
In the subsequent discussion the  initial condition  
\begin{equation}
  u(x,0) = 0.45(1+ 0.01 \cos{\mu x}), \qquad \mu=\frac{2 \pi}{L}, \, L = 4\sqrt{2} \pi,
  \label{IC-A}
\end{equation} 
will often be used to
illustrate the salient features of the various models.

\subsubsection{Linear damped HONLS  model: \texorpdfstring{$\epsilon = 0.05, \Gamma >0, \nu = 0, \delta = 0$}{LDR}}
\label{results_LD}
In the linear damped (LDHONLS)
framework the individual Fourier modes experience uniform damping,
as described by equations \eqref{FdPdt} - \eqref{FdEdt}, and $k_m$ remains constant in time. Consequently, frequency downshifting is not expected.

As an example,  consider the solution of the LDHONLS equation with $\Gamma = 0.01$ and  initial condition $u_0 = 0.5(1 + 0.1 \cos\mu x)$, $L = 4\sqrt{2}\pi$,
in the two-unstable mode regime.
Figure~\ref{LD_2UMR_good}
shows the evolution of (a) the strength $S(t)$,  (b) the  spectral center $k_m$,
(c) the spectral peak $k_{peak}$, and (d) the amplitude of the Fourier modes $|\hat u_{k}|$,
$k = 0, \pm 1, \pm 2$.
\begin{figure}[!ht]
  \centerline{
    \includegraphics[width=.3\textwidth]{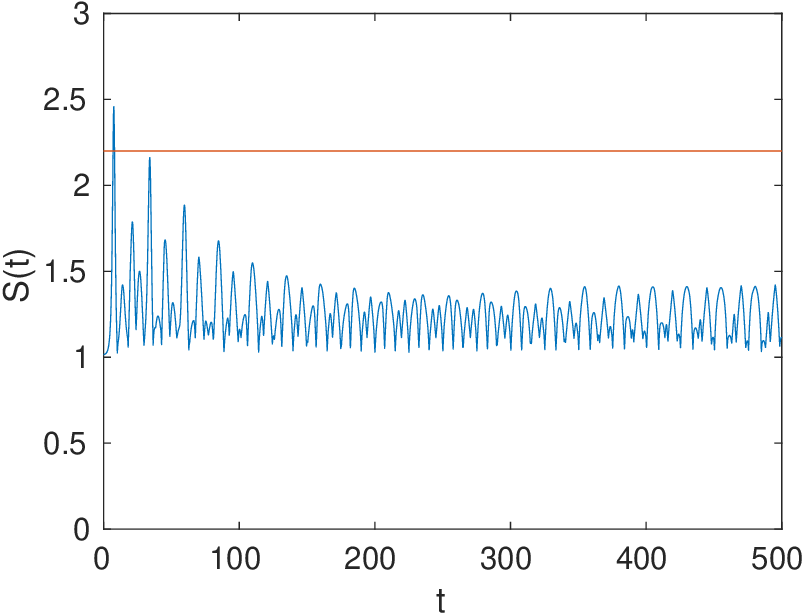}\hspace{12pt}
    \includegraphics[width=.3\textwidth]{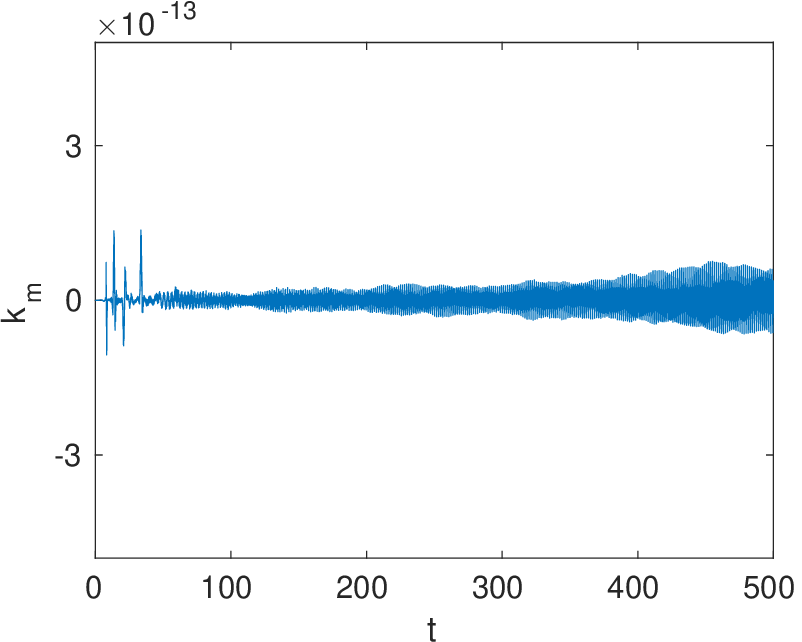}
    }
  \centerline{
    \includegraphics[width=.3\textwidth]{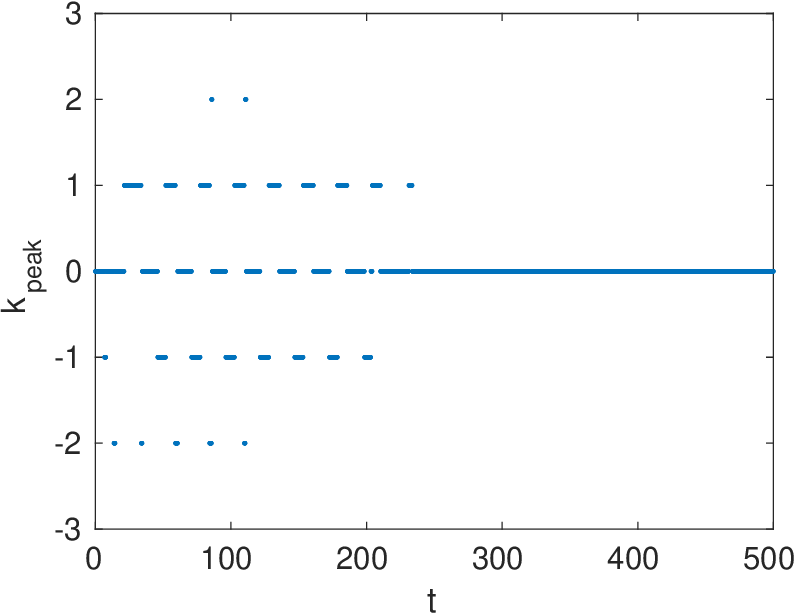}\hspace{12pt}
    \includegraphics[width=.3\textwidth]{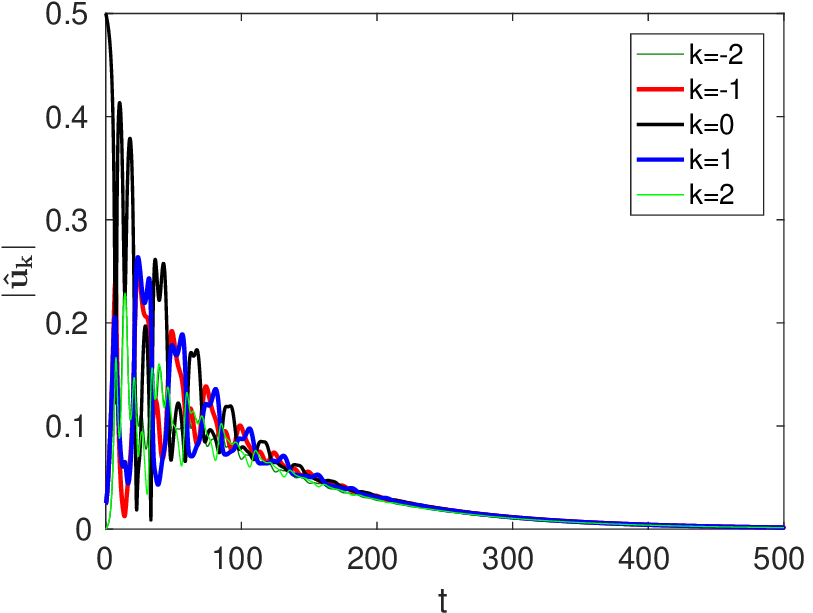}
  }
  \caption{LDHONLS equation with $\Gamma = 0.01$ and initial data $u_0 = 0.5(1 + 0.1 \cos\mu x)$, $L = 4\sqrt{2}\pi$. Numerical evolutions of 
  (a) the strength $S$,  (b) $k_m$,
(c)  $k_{peak}$ and (d) $|\hat u_{k}|$.}
 \label{LD_2UMR_good}
  \end{figure}
The strength plot shows a single rogue wave developing early in the experiment.
Within integrator accuracy of $\mathcal{O} (10^{-13})$, the spectral center is  constant at $k_m = 0$.
Temporary shifting in $k_{peak}$ occurs until $t \approx 230$ when  $k_{peak}$ settles at $k = 0$. The Fourier mode plot shows that  $|\hat u_0|$ (black curve) is dominant
with extremely
small differences in $|\hat{u}_k|^2 -|\hat{u}_{-k}|^2$. The two spectral diagnostics $k_m$ and $k_{peak}$ consistently indicate that  frequency downshifting does not occur.

Even though $\displaystyle\frac{d k_m}{dt} = 0$ for the LDHONLS, in some cases the spectral peak $k_{peak}$ may shift downwards.
 Figure~\ref{LD_out} shows the evolution of  (a) the strength $S$,  (b) the spectral center $k_m$,
(c) the spectral peak $k_{peak}$, and (d) $|\hat u_{k}|$ for
$\Gamma = 0.01$ and initial data \rf{IC-A}.
\begin{figure}[!ht]
  \centerline{
    \includegraphics[width=.3\textwidth]{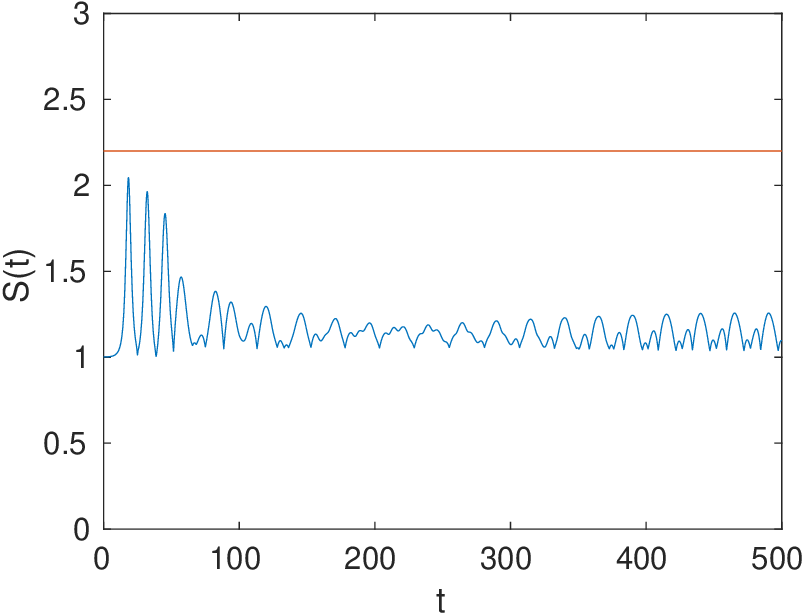}\hspace{12pt}
    \includegraphics[width=.3\textwidth]{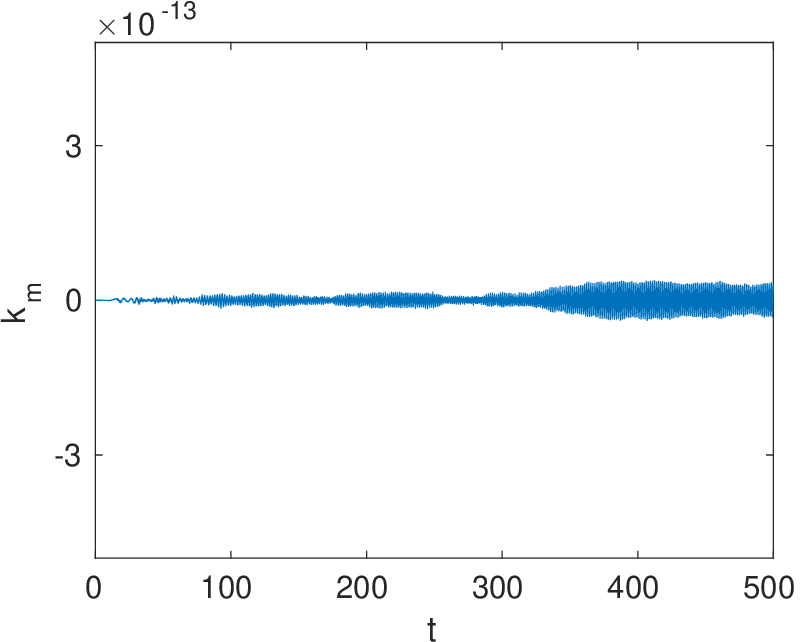}
  }
    \centerline{
    \includegraphics[width=.3\textwidth]{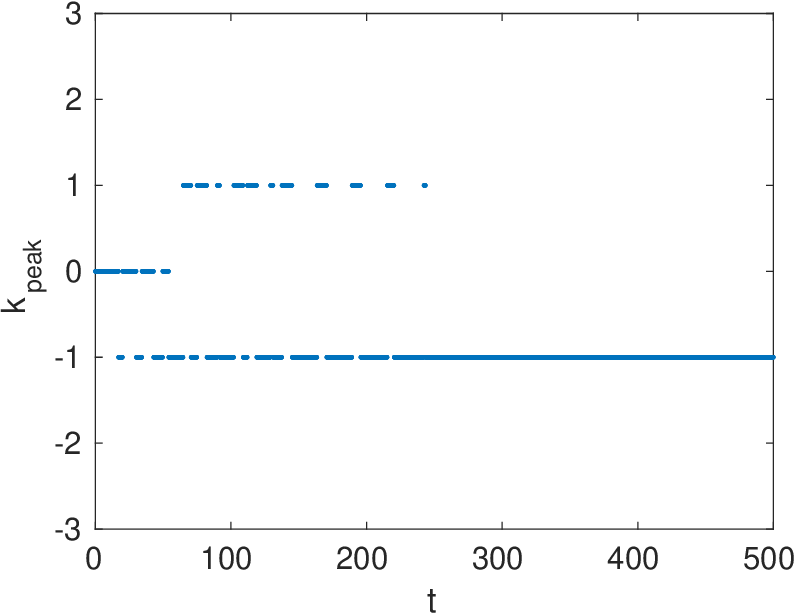}\hspace{12pt}
    \includegraphics[width=.3\textwidth]{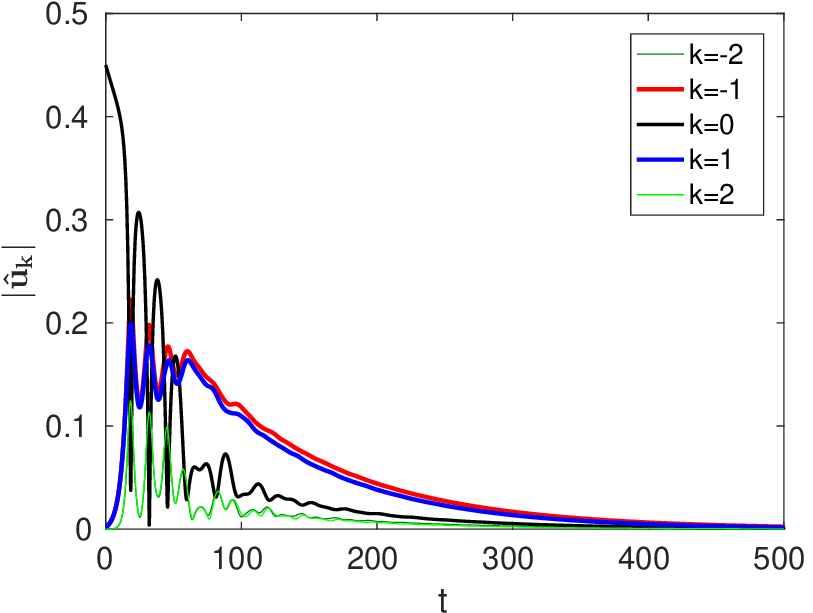}
  }    
  \caption{LDHONLS equation with  $\Gamma = 0.01$ and initial data (\ref{IC-A}).  Numerical evolution of 
  (a) the strength S,  (b) $k_m$,
    (c)  $k_{peak}$ and (d) $|\hat u_{k}|$.}
  \label{LD_out}
  \end{figure}
 The Fourier mode plot shows that, although $|\hat u_{-1}|$ (red curve) dominates,
$|\hat u_{\pm 1}|$ and $|\hat u_{\pm 2}|$ are nearly identical thus preserving
$k_m = 0$. Downshifting  does not occur in terms of the spectral mean.
However, downshifting is observed in the spectral peak sense as
 $k_{peak}$ shifts down to $k = -1$
at $t \approx 220$. 
In this situation, since the  working criterion is not satisfied, it is viewed that downshifting does not occur.

\subsubsection{Viscous HONLS model: \texorpdfstring{$\epsilon = 0.05, \Gamma >0, \nu = 1, \delta = 0$ and
$\delta = 42\Gamma$ (i.e. $\epsilon\delta = 2.1\Gamma$)}{VDR}}
\label{visc_results}
In this subsection we 
discuss our findings on the mechanism for downshifting and the
time of permanent downsift $t_{DS}$ 
for the viscous models  for two specific values of $\Gamma$.
In the next subsection,   the ensemble results are presented in
Figures~\ref{time_ds} (a) - (b)  which confirm the  downshifting mechanism and its relation to
$t_{DS}$ across all given values of $\Gamma$.

Consider the solutions of
the vHONLS  and  dvHONLS 
equations,  with $\Gamma = 0.00275$ and  $\Gamma = 0.005$, for 
initial data \rf{IC-A}. 
Figure~\ref{vHONLS_both_km_00275}
shows the evolution of  (a) the energy $E$,  (b) the momentum $P$, and
(c) the spectral center $k_m$ for   $ 0 \leq t \leq 500$. The results for the vHONLS ($\delta = 0$) and dvHONLS ($\delta = 42\Gamma$)  equations are presented for the two values of $\Gamma$.
Following Proposition (\ref{prop2p2}),  the diffusion coefficient $\delta = 42\Gamma$ (i.e. $\epsilon\delta = 2.1\Gamma$) is chosen to ensure that the additional
instabilities arising from the viscous term are stabilized in the dvHONLS equation.
The graphs for $\Gamma = 0.00275$ are rendered with solid lines, while the
results for $\Gamma = 0.005$ are illustrated using dashed lines.

\begin{figure}[!ht]
  \centerline{
    \includegraphics[width=.3\textwidth]{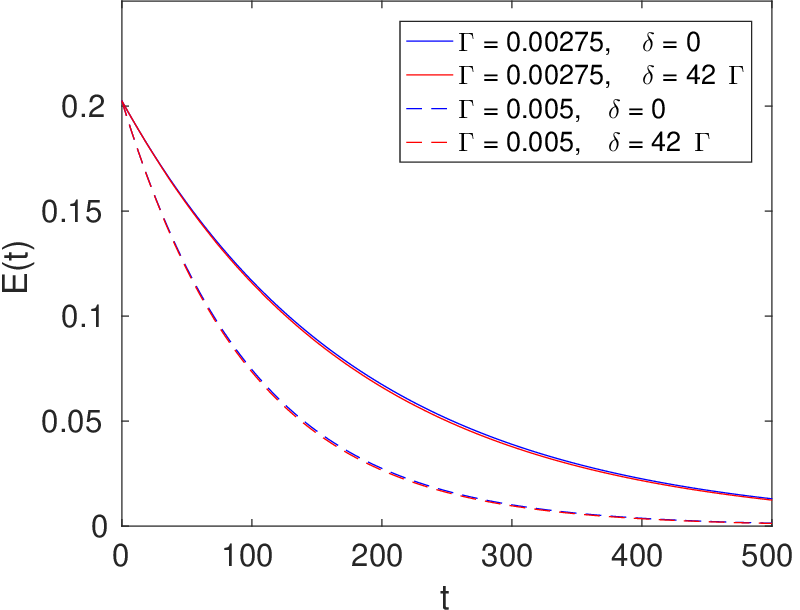}\hspace{12pt}
    \includegraphics[width=.32\textwidth]{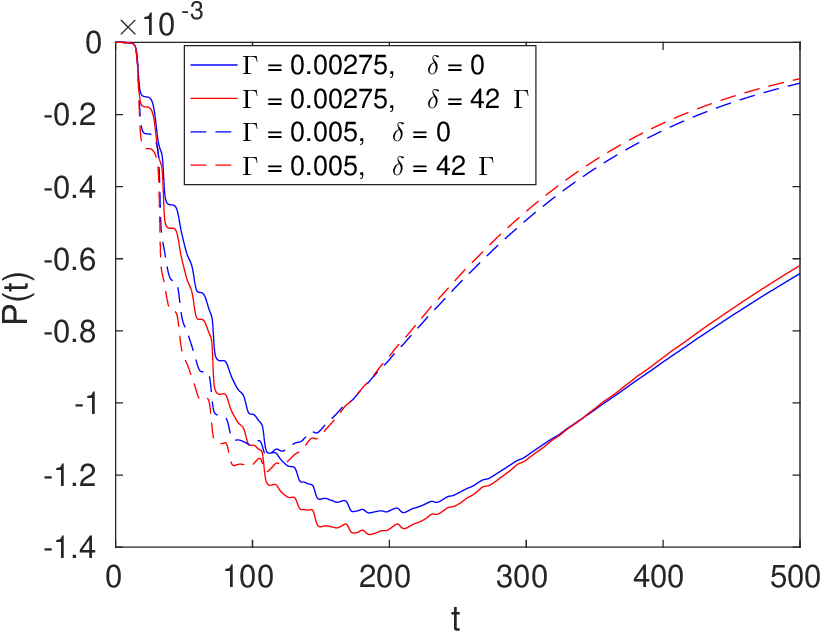}\hspace{12pt}
    \includegraphics[width=.32\textwidth]{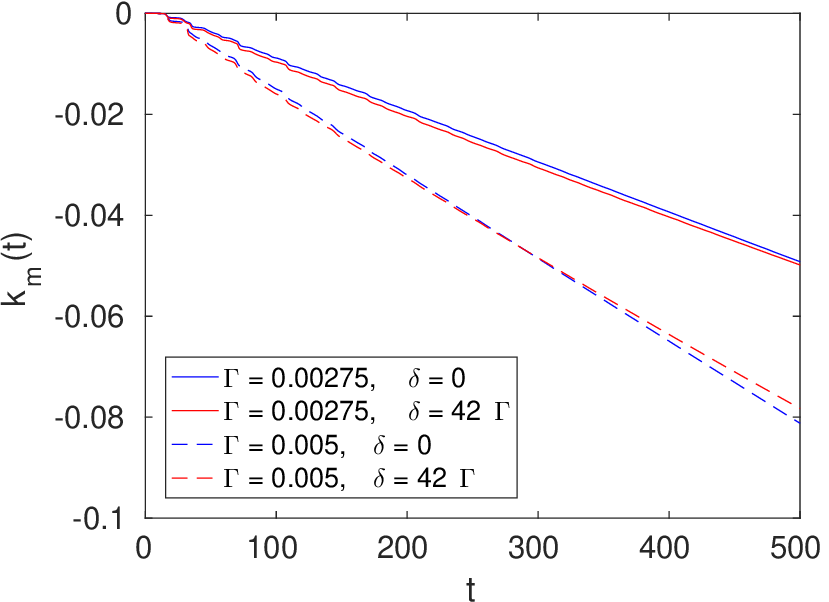}
  }
  \caption{vHONLS and dvHONLS equations with $\Gamma = 0.00275$ and $\Gamma = 0.005$,
    for initial data (\ref{IC-A}). Evolutions of  a) the energy E, b) the momentum P, and c) the spectral center $k_m$.}
  \label{vHONLS_both_km_00275}
\end{figure}

The energy exhibits an  exponential-type decay and the spectral center $k_m$ is strictly decreasing for
both the vHONLS and dvHONLS models. This is 
consistent with the behavior described by equations \eqref{derivE} and
(\ref{derivkm}). The evolution of $P$ in the vHONLS model is 
in striking contrast with the exponential decay of the
momentum in the LDHONLS model (see equation~\eqref{derivP} with $\nu=\delta=0$.)
In Figure \ref{vHONLS_both_km_00275}(b), we observe (at a  macroscopic level)
that the graph of
$P$ transitions from 
concave down to concave up and then back to  concave down.
$P$ achieves its global minimum at a time $t_{P_{min}}$, which is 
followed by a decrease in the magnitude  $|P|$. For example, 
$t_{P_{min}} = 185$  for the vHONLS equation with $\Gamma = 0.00275$.
\smallskip

Although equation (\ref{derivkm}) predicts downshifting  in the spectral mean sense, it  is not immediately apparent from equation (\ref{derivP}) whether and when there will be  an accompanying permanent downward shift in $k_{peak}$. 
In Figures~\ref{KPE_woD_withD_00275} (a) -- (b) the time of permanent downshift as determined by $k_{peak}$ is $t_{DS} = 265$ for the vHONLS equation and a slightly larger value for the dvHONLS equation.
Figure~\ref{KPE_woD_withD_00275} (c) shows the Fourier
mode amplitudes $|\hat u_{k}|$ for $k = 0, \pm 1, \pm 2$ for the vHONLS.   The Fourier mode
$\hat u_{-1}$, whose amplitude is given by the red curve, is
clearly the dominant mode. We note that the difference
$|\hat u_{k}|$  - $|\hat u_{-k}|$ is $\mathcal{O} (10^{-2})$ for $k = 0, \pm 1, \pm 2$.

Spectral activity,  with shifts in $k_{peak}$ between the sideband modes, occurs for $t > t_{DS}$ in both the vHONLS and dvHONLS (see Figure~\ref{KPE_woD_withD_00275} (a) - (b)). To characterize the regime when the shifting between the sidebands has ceased, we introduce the following:
\begin{dfn} Let $t_{peak}$ be the time s.t. $k_{peak}(t)=c$ for all $t\geq t_{peak}$, where $c$ is a constant.
\end{dfn}
In the above example, $t_{peak}$ is greater for the vHONLS in comparison with the dvHONLS. 

The results provided in Figures~\ref{vHONLS_both_km_00275} -- \ref{KPE_woD_withD_00275}
indicate  that the decay in $E$, coupled with the increase in the magnitude of
the momentum $|P|$ during the early stages of the wave evolution, results in a predominantly downward shift in $k_{peak}$.

\begin{figure}[!ht]
  \centerline{
    \includegraphics[width=.3\textwidth]{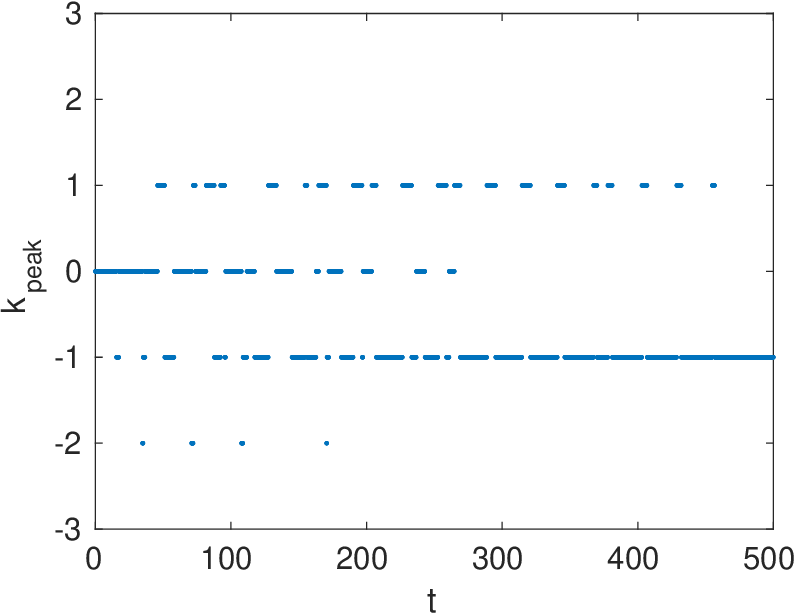}\hspace{12pt}
    \includegraphics[width=.3\textwidth]{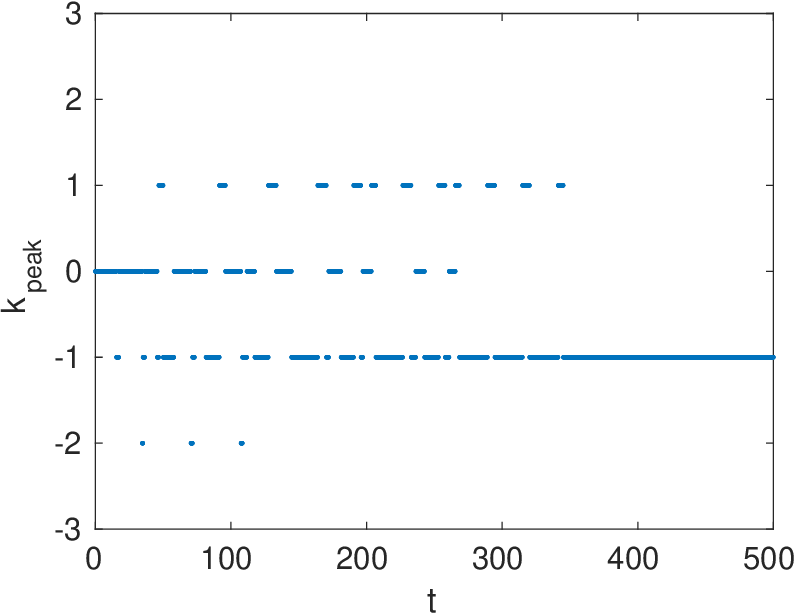}\hspace{12pt}
    \includegraphics[width=.3\textwidth]{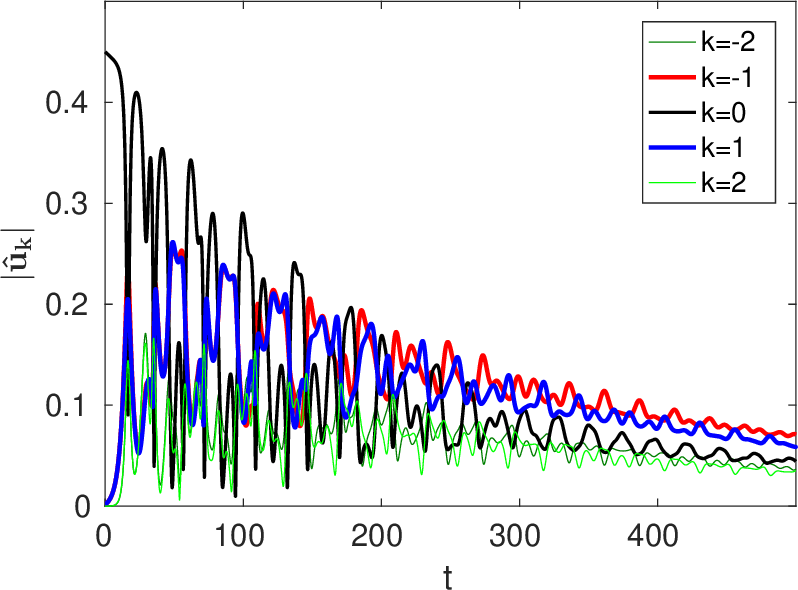}
     }
  \caption{Evolution of the spectral peak $k_{peak}$ for the (a)   vHONLS and (b) dvHONLS equations with  $\Gamma = 0.00275$ and $u_0 = 0.45(1 + 0.01 \cos\mu x)$. (c) Evolution of the Fourier modes $|\hat u_{k}|$ for the vHONLS.}
  \label{KPE_woD_withD_00275}
\end{figure}

Several other features of $P$ and its derivatives
are relevant to  downshifting in the vHONLS.
Figures~\ref{PanddPdt} (a) -- (c) show the plots of $P$,
$P'$,
and $P''$ for $\Gamma = 0.00275$ and 
initial data \ref{IC-A}. The
blue and red  curve segments mark the time intervals when $k_{peak}=0$ and $k_{peak}\neq 0$, respectively.
We find that the time $t_{P_{min}}$ of the global minimum of $P$
is an important {\em precursor} to the time $t_{DS}$ of 
permanent downshift, i.e., $t_{DS} > t_{P_{min}}$. 
Large rapid fluctuations in $P'$
and $P''$ are observed initially, which 
decay substantially 
in magnitude as time evolves, with {\em most} of the decay occuring for $ 0 < t <  t_{P_{min}}$.
Permanent downshift  occurs
when the magnitude $|P|$ is decreasing 
and the derivatives
$P'$ and $P''$ become small enough that changes in $P'$ are no longer sufficient for $k_{peak}$ to upshift back.

\smallskip

\begin{observe}
  For the particular example under consideration ($\Gamma = 0.00275$),  the vHONLS and dvHONLS evolutions exhibit these important features: 1)
  $t_{DS} > t_{P_{min}}$; \,
  2) there are rapid, step-like changes in the momentum during the transition from $k_{\text{peak}} = 0$ to $k_{\text{peak}} \neq 0$ (and vice versa), which
 occur mostly before $t_{P_{min}}$and which  diminish in magnitude as time evolves; and 3) permanent downshift occurs when
  $|P|$ is strictly decreasing and once  $P'$ and $P''$ become  sufficiently small. 
\end{observe}
\begin{figure}[!ht]
  \centerline{
    \includegraphics[width=.3\textwidth]{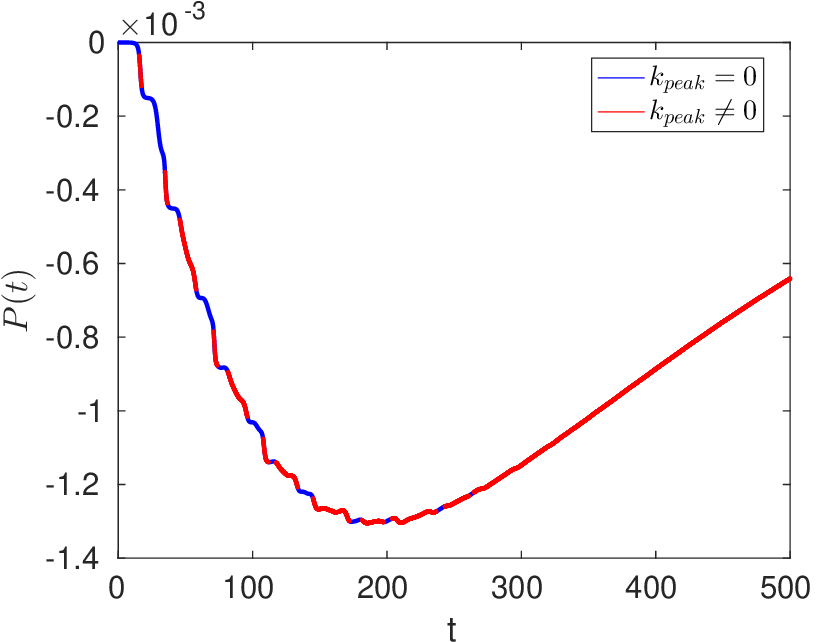}\hspace{12pt}
   \includegraphics[width=.3\textwidth]{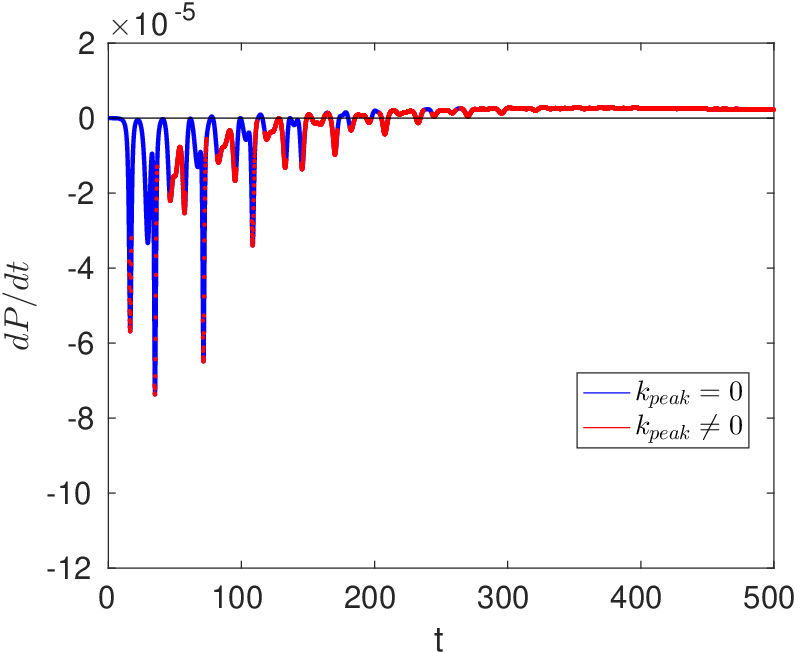}\hspace{12pt}
   \includegraphics[width=.3\textwidth]{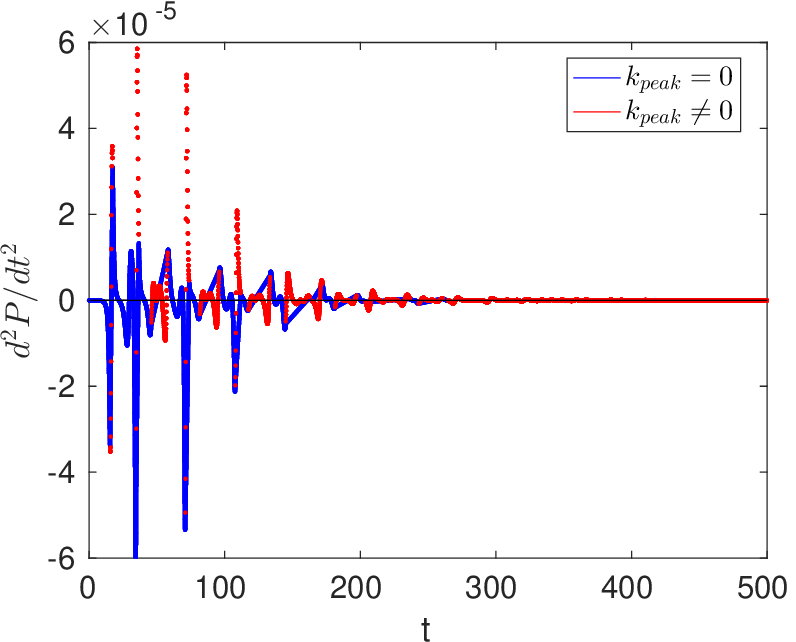}
  }
  \caption{Evolutions of (a) the momentum $P$ and its derivatives (b)  $P'$
    and (c)  $P''$ for $\Gamma = 0.00275$ and initial data \ref{IC-A}. Blue and red portions of the graphs mark the time intervals when $k_{peak}=0$ and $k_{peak}\neq 0$, respectively. }

  \label{PanddPdt}
\end{figure}

In order to provide an analytical explanation of the observed behavior of $P$, we examine the Fourier representation of $P'$
given by equation (\ref{FdPdt}).  The  initial condition used in the numerical experiments is symmetric with 
 $P(0) = 0$. The presence of the  negative $\mathcal{O} (\epsilon\Gamma)$-term in equation \eqref{derivP} leads to $P'(t)<0$ for some  time interval $0 < t < t_1$.
As discussed in Corollary \ref{negk}, the introduction of viscosity in the vHONLS model induces long-term instabilities primarily in the negative Fourier modes.
In the asymptotic regime,   the negative modes $|\hat{u}_{-k}|$ dominate the evolution of the solution of equation (\ref{FdPdt}), leading to  $P'(t) > 0$ for large $t$.

A rough estimate of the  time where the global minimum of $P$ occurs for the vHONLS equation can be obtained from expression \eqref{FdPdt} with $\delta = 0$ by only retaining  the modes in the series up to $k = \pm 1$. The resulting truncation gives
\begin{equation}
   \frac{dP}{dt} \approx -2\Gamma \left(|\hat u_1|^2(1+2\epsilon) - |\hat u_{-1}|^2(1-2\epsilon)\right).
  \label{Pdottrunc}
  \end{equation}
Approximating $|\hat{u}_{\pm 1}|$ with its asymptotic expression for large times (see  Section~\ref{largetime}), we substitute  $|\hat{u}_{\pm 1}|^2 \sim e^{\mp 2\gamma t}$ 
into
equation (\ref{Pdottrunc}) and solve  for $P'(t) = 0$ to obtain
\[e^{4\gamma t} \approx  \frac{1+2\epsilon}{1-2\epsilon},
\]
which provides the following estimate for the time of global minimum of $P$:
\begin{equation}
  t_{P_ {min}} \approx \frac{1}{4\gamma}\ln\frac{1+2\epsilon}{1-2\epsilon}
  \approx  \frac{1}{2\Gamma}\left(1 + \frac{4}{3}\epsilon^2\right),
  \label{tpmin}
\end{equation}
where we used $\gamma = 2\epsilon \Gamma$.
\begin{rem}
When $\Gamma = 0.00275$ we compute  $t_{P_{min}} \approx  181.6$, in good  agreement with the value  for $t_{P_{min}} = 185$ obtained in the 
numerical experiment.
\end{rem}

We observe that varying $\Gamma$ leads to qualitatively similar  graphs for E, P, and $k_m$.
Comparing  the graphs of $P$
in Figure \ref{vHONLS_both_km_00275} (b), we find that the global minimum of  $P$ occurs earlier for the larger value $\Gamma = 0.005$. In the following ensemble results for the viscous models we find $t_{P_{min}}$ is a decreasing function of $\Gamma$,
in agreement  with the analytical approximation given by~\eqref{tpmin}.


\subsubsection{Ensemble results for the viscous HONLS model: \texorpdfstring{$\epsilon = 0.05$, $\Gamma >0$, $\nu = 1$, $\delta =0$ and $\delta = 42\Gamma$ (i.e., $\epsilon\delta = 2.1\Gamma$)}{ER}}
\label{ensemble}
An ensemble of experiments for the vHONLS and dvHONLS models was carried out 
using initial data (\ref{IC-A}) and varying $\Gamma$ within the range
$2.5 \times 10^{-4} \le \Gamma \le 10^{-2}$ with increment $\Delta\Gamma = 2.5\times 10^{-4}$.
 The purpose of the ensemble experiments is to investigate the dependence of the downshifting mechanism as well as of quantities such as (i) $t_{P_{min}}$, (ii) $t_{DS}$, and (iii) $t_{peak}$  on the damping parameter $\Gamma$. The relation between frequency downshifting and rogue wave formation is also examined. Note: For the viscous HONLS models, equation~ (\ref{key}) predicts downshifting  in $k_m$ for all $\Gamma$ and $\epsilon$ considered.

Figure \ref{time_ds} (a)  shows $t_{DS}$ and  $t_{P_{min}}$ as functions of   $\Gamma$ for the vHONLS model using initial data (\ref{IC-A}).  The corresponding graphs for the  the dvHONLS model
are given in Figure \ref{time_ds} (b).  In Figure \ref{time_ds} (c),
the dependence of  $t_{peak}$  on $\Gamma$ is  given for both the vHONLS and dvHONLS models.

\subsubsection*{Results of Ensemble Experiments}

\begin{itemize}
\item[i)] Frequency downshifting  occurs in the solutions  of
  both the vHONLS and dvHONLS equations across all initial conditions in the two-unstable modes regime  and for all given values of $\Gamma$.
  \item[ii)]        
We find that $t_{P_{min}}$ serves as an important precursor to $t_{DS}$
    for a broad range of  $\Gamma$ values.
    Specifically, in Figures \ref{time_ds} (a) and (b),  $t_{P_{min}} < t_{DS}$
    for $1.25 \times 10^{-3} \leq \Gamma \leq 8 \times 10^{-3}$. 
    
    \item[iii)] 
    Downshifting occurs via the following mechanism: 
    the initial decay of the energy $E$ and step-like increases in the  magnitude $|P|$ of the momentum  result in a predominantly downward shift in $k_{peak}$.  The global minimum of $P$ is reached at time $t_{P_{min}}$, followed by a decrease in $|P|$. Viscous damping  induces large rapid fluctuations in
    ~$P'$ and  $P''$, 
    which decay in magnitude as time evolves, with most of their decay occuring
    for $ 0 < t < t_{P_{min}}$.
The time of downshift $t_{DS}$ occurs  once $P''$ has become small enough
that changes in $P'$ are no longer sufficient for $k_{peak}$ to upshift back.

\item[iv)]  As $\Gamma$ increases, the rapid large fluctuations
  in~$P'$ and  $P''$
  diminish over progressively shorter time intervals. As a result, $t_{DS}$ approaches  $t_{P_{min}}$. At 
    $\Gamma = 0.008$,   we have $t_{P_{min}} \approx t_{DS}$.
For $0.008 \leq \Gamma \leq 0.01$, the downshifting mechanism follows the same pattern as in (iii) with  $t_{DS}$ occuring slightly 
before $t_{P_{min}}$ (the two times  very close).
 Figures~\ref{largeG} (a) - (c) show the evolutions of $P$ and its derivatives
for the vHONLS with $\Gamma = 0.01$ (the upper end of the parameter range) using initial data (\ref{IC-A}).

\item[v)] Both the time of the global minimum of
$P$ and the  time of downshift predominantly {\em decrease} as the value of $\Gamma$ is {\em increased}.  The solid red and blue curves give the best fit to the data for  $t_{Pmin}$ and $t_{DS}$, respectively.  To illustrate this for the vHONLS equation,  we find $t_{Pmin} =  1/2\Gamma + 8.8$ and
$t_{DS} =   0.987/\Gamma - 57.7$
 where
 $\Gamma \in [1,10]\times 10^{-3}$. Notice that there is a very good agreement between the analytical approximation of $t_{Pmin}$ given by equation~\eqref{tpmin} and the best fit from the numerical data.
 
\item[vi)]    
 Figure \ref{time_ds}  (c) shows that for all values of $\Gamma$,
 $t_{peak}$ is greater  for the vHONLS model in comparison with the dvHONLS model. This suggests increased spectral activity for the viscous case. Furthermore, the plot reveals that $t_{peak}$ also  decreases when $\Gamma$ increases.
\end{itemize}

\begin{figure}[!ht]
  \centerline{
\includegraphics[width=.3\textwidth]{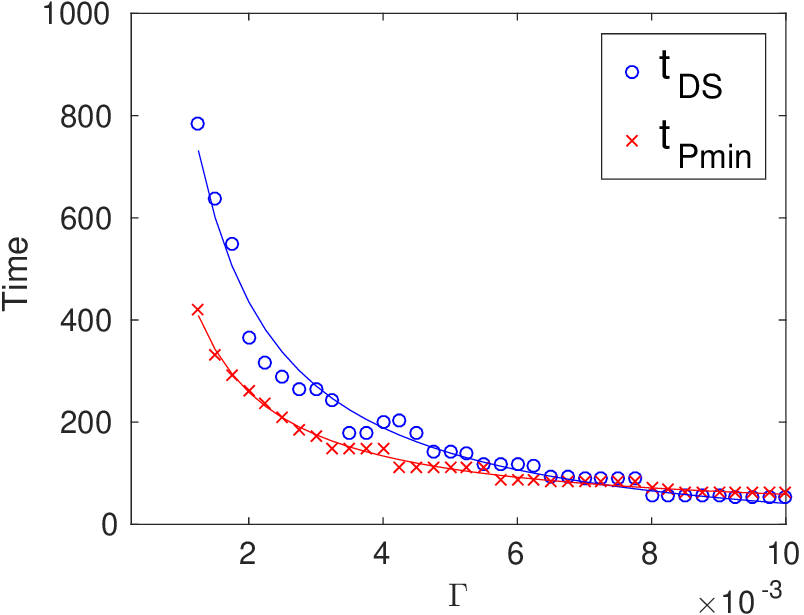}\hspace{12pt}

\includegraphics[width=.3\textwidth]{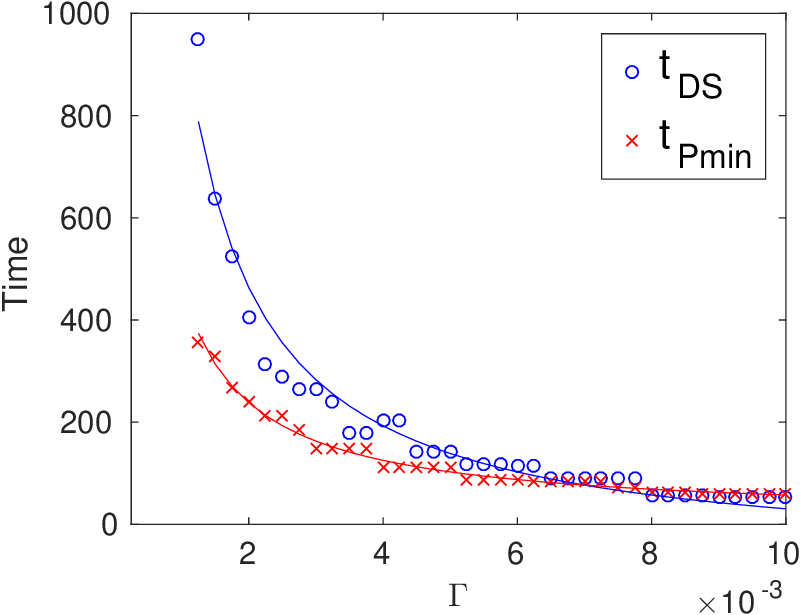}\hspace{12pt}
   \includegraphics[width=.3\textwidth]{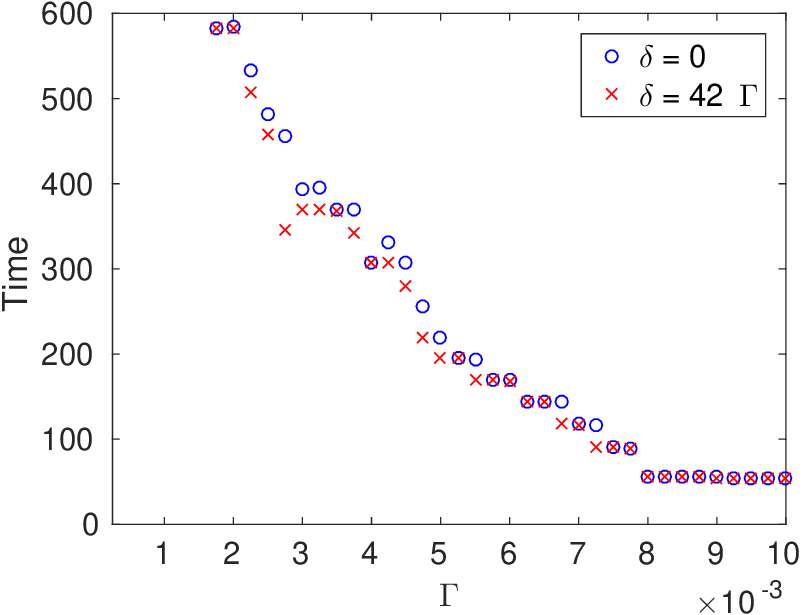}
  }
  \caption{Time of permanent downshift $t_{DS}$, and  time of global minimum of $P$, $t_{P_{min}}$,   for the a) vHONLS and b)  dvHONLS equations. Initial data (\ref{IC-A}) are used.  
    c)  Time of  $t_{peak}$  as function of $\Gamma$
 for the vHONLS ($\delta = 0$)  and dvHONLS ($\delta = 42\Gamma$) models.}
  \label{time_ds}
\end{figure}

\begin{figure}[!ht]
  \centerline{
    \includegraphics[width=.3\textwidth]{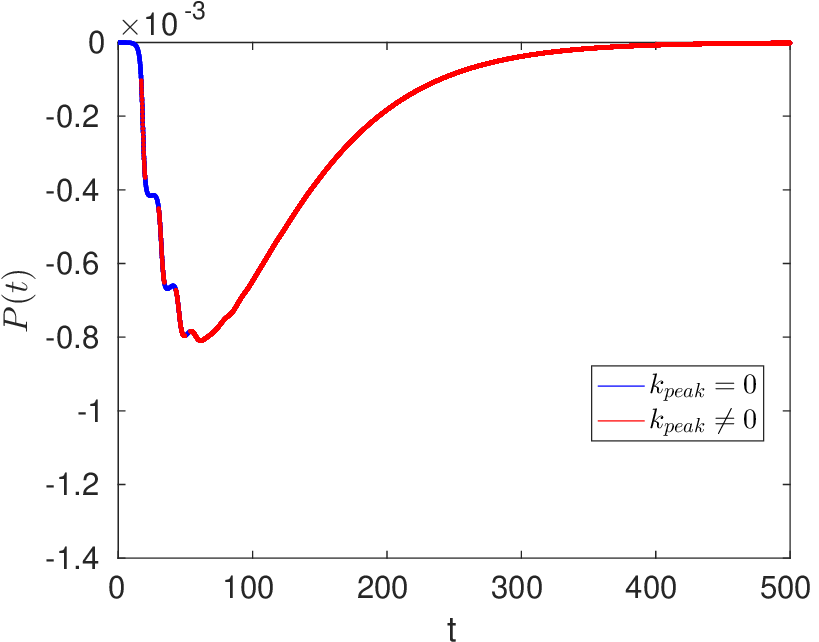}
    \includegraphics[width=.3\textwidth]{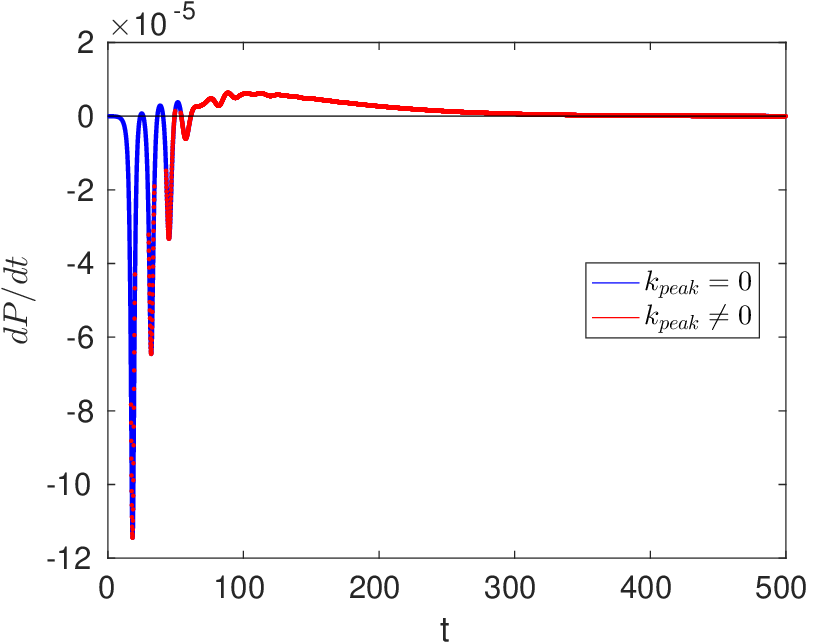}
    \includegraphics[width=.3\textwidth]{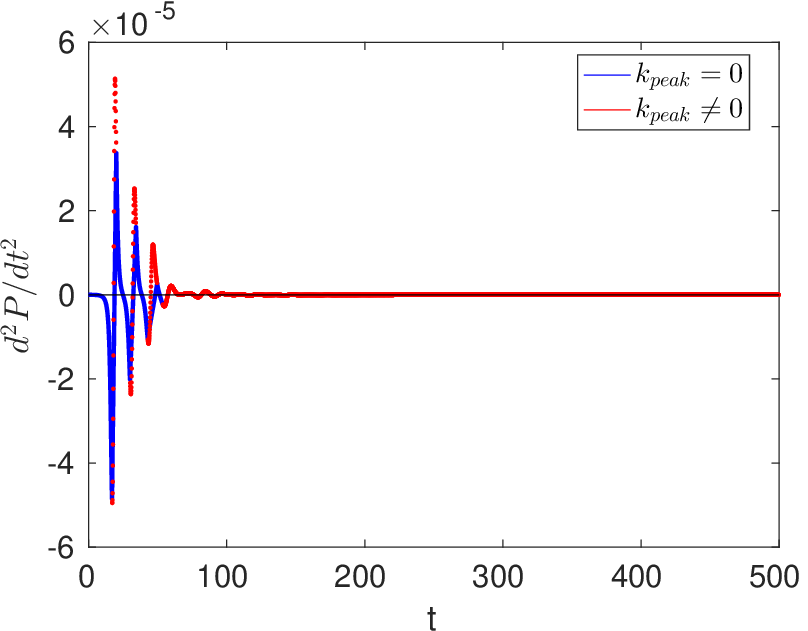}
  }
  \caption{Evolutions of (a) the momentum $P$ and its derivatives (b)  $P'$
    and (c)  $P''$ for $\Gamma = 0.01$ and initial data (\ref{IC-A}). Blue and red portions of the graphs mark the time intervals when $k_{peak}=0$ and $k_{peak}\neq 0$, respectively. Here $t_{P_{min}} = 61.3$ and
    $t_{DS} = 55.8$.}
  \label{largeG}
\end{figure}

\subsubsection*{Rogue wave activity}

 \begin{itemize}
 \item[i)]
Figures \ref{wave_st} (a)  -- (c) show the wave strength for solutions of
the LDHONLS, vHONLS, and   dvHONLS models, respectively, for $0 < t < 500$, $\Gamma = 0.00275$  and 
initial data \ref{IC-A}.
In this example  more rogue waves occur in the vHONLS and dvHONLS evolutions than in the LDHONLS evolution. 

\item[ii)]
 Figures \ref{NRWs} (a) -- (b) compare the number of rogue waves obtained in two sets of experiments: a) LDHONLS versus vHONLS evolution, and b) LDHONLS versus dvHONLS evolution. The
 comparison is conducted  over the time interval $0 < t < 500$, with $\Gamma$ as the varying parameter and initial data \ref{IC-A}.
 We find that, for all values of $\Gamma$, as many or more rogue waves develop in the vHONLS and dvHONLS evolution as compared with the LDHONLS evolution.

\item[iii)]
 In Figure \ref{NRWs} (c), a moving average of the time of permanent downshift is presented as a function of the maximum strength for a set of vHONLS (red curve) and dvHONLS (blue curve) experiments. Here  $\Gamma$ is varied in the interval $2.5 \times 10^{-4} \le \Gamma \le 10^{-2}$ and the initial condition is fixed as in \eqref{IC-A}. The resulting graphs exhibit an upward trend, indicating that for both the vHONLS and dvHONLS evolutions, the time of permanent downshifting is directly related to the maximum strength attained during the simulation. In other words, the greater the wave strength, the later the permanent downshift will occur.

\item[iv)]
 Figures~\ref{time_tlrw}   (a) - (b) provide the time of permanent downshift (marked by an 'x') and the time of the last rogue wave (marked by a 'box') for the (a) vHONLS and (b) dvHONLS models as functions of $\Gamma$, again using
 initial data \ref{IC-A}.
 Notably, our findings indicate that rogue waves  do not occur after permanent downshifting. 
\end{itemize}

\begin{figure}[!ht]
  \centerline{
    \includegraphics[width=.3\textwidth]{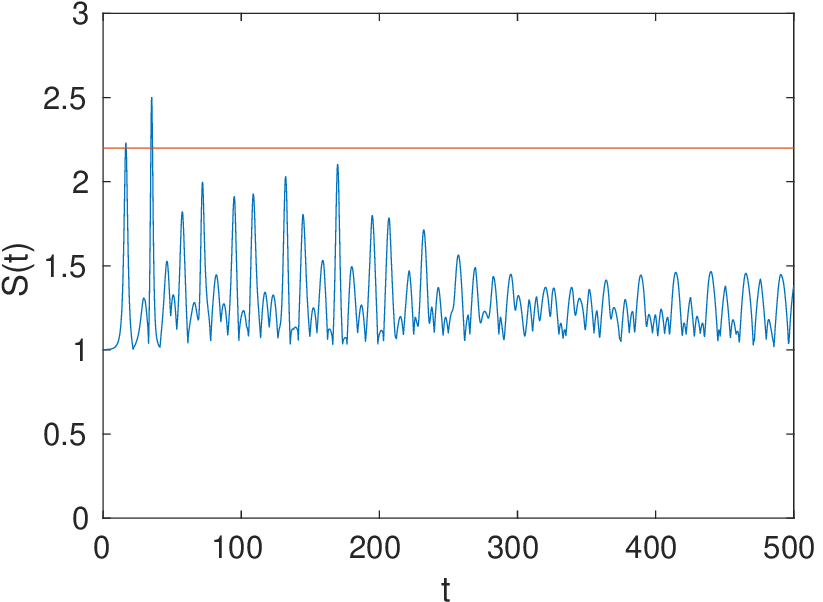}\hspace{12pt}
    \includegraphics[width=.3\textwidth]{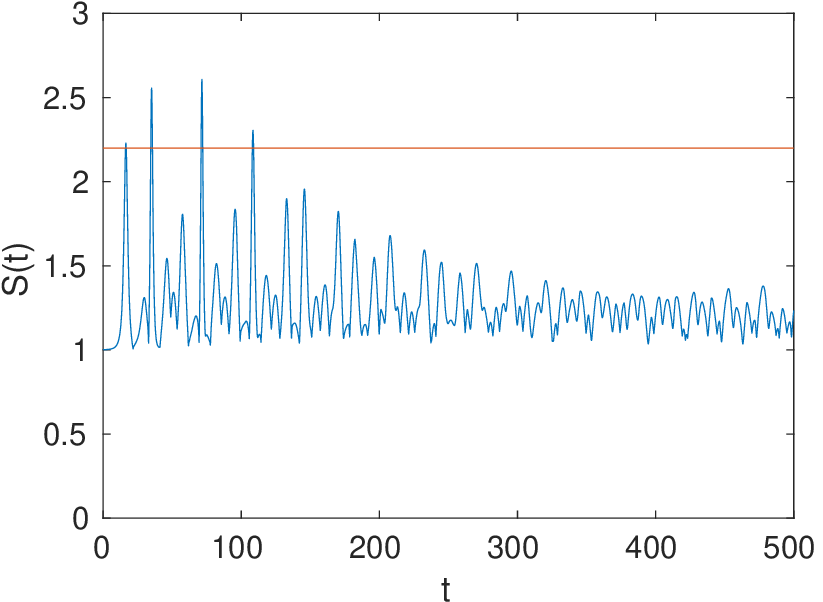}\hspace{12pt}
    \includegraphics[width=.3\textwidth]{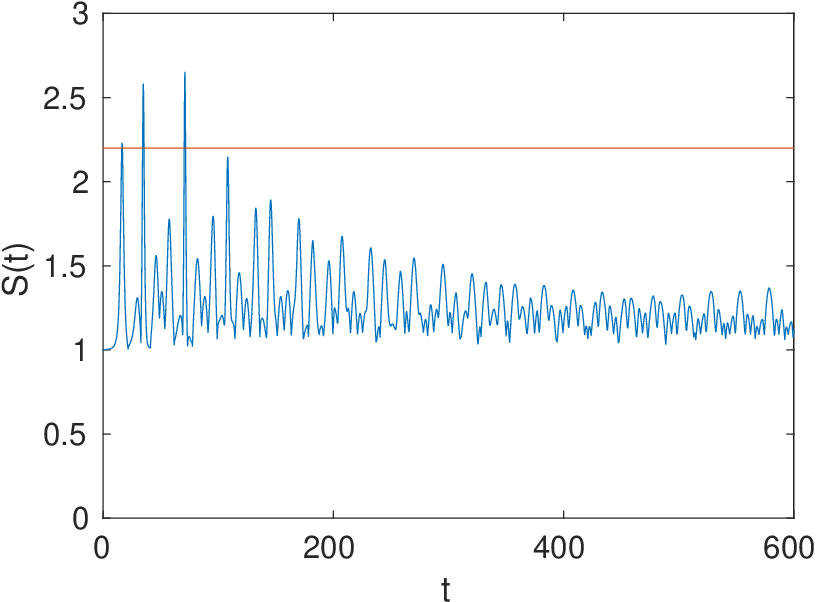}
  }
  \caption{Wave strength for  the a) LDHONLS, b) vHONLS, and  c) dvHONLS evolutions, for $\Gamma = 0.00275$ and initial data (\ref{IC-A}).}
  \label{wave_st}
  \end{figure}

\begin{figure}[!ht]
  \centerline{
    \includegraphics[width=.3\textwidth]{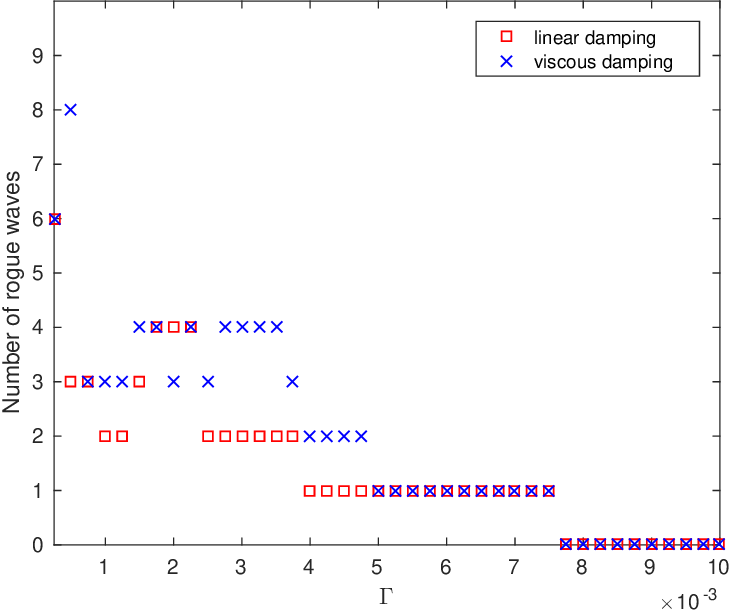}\hspace{12pt}
    \includegraphics[width=.3\textwidth]{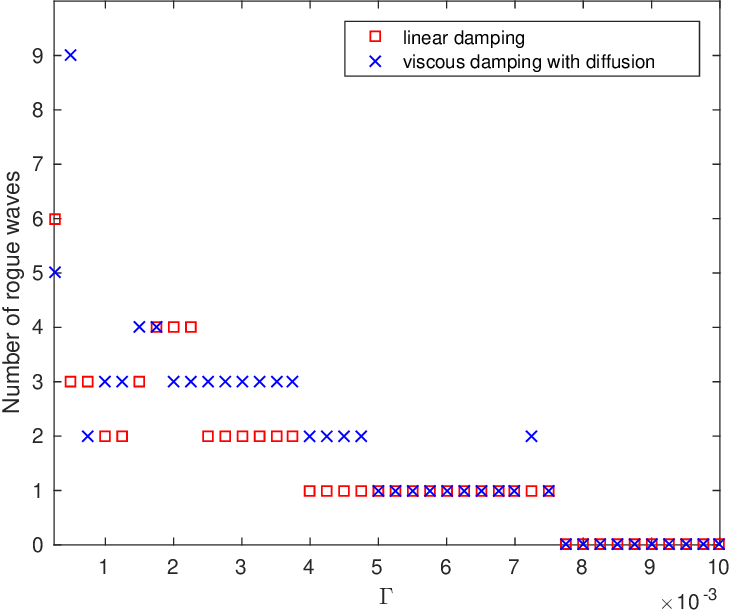}\hspace{12pt}
        \includegraphics[width=.34\textwidth]{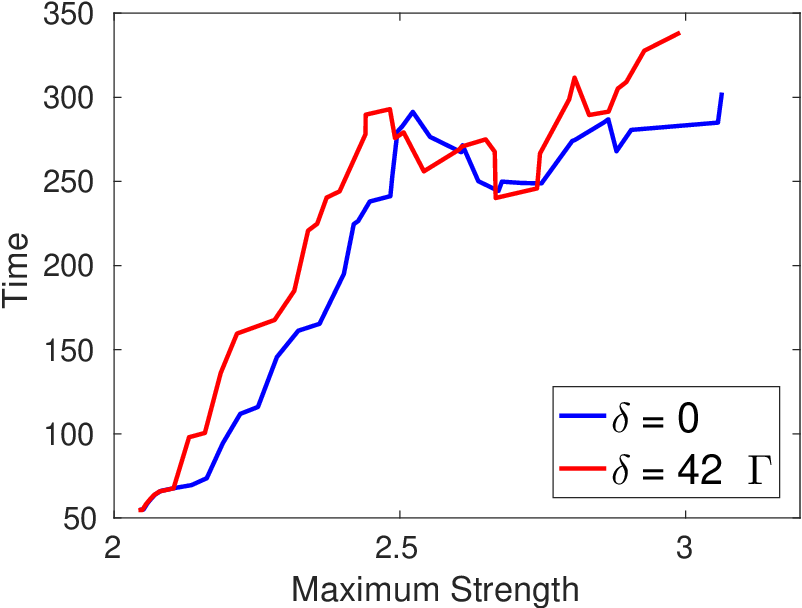}
  }
  \caption{Number of rogue waves as a function of $\Gamma$ for
    a) LDHONLS vs vHONLS evolutions, and b) LDHONLS vs dvHONLS evolutions . c)  Moving mean value of time of downshifting as a function of the maximum strength.  $\epsilon = 0.05$, $2.5 \times 10^{-4} \le \Gamma  \le 10^{-2}$.  
}
  \label{NRWs}
  \end{figure}

\begin{figure}[!ht]
  \centerline{
    \includegraphics[width=.31\textwidth]{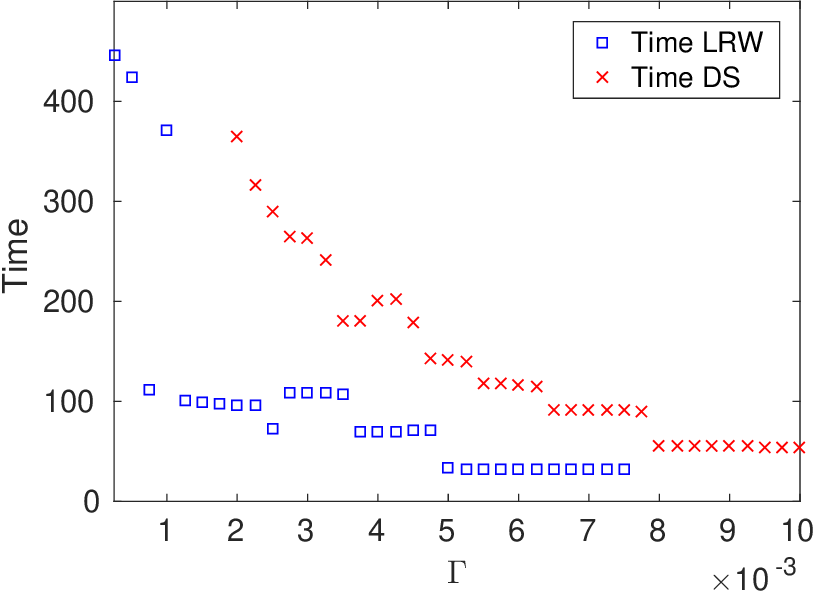}\hspace{14pt}
    \includegraphics[width=.3\textwidth]{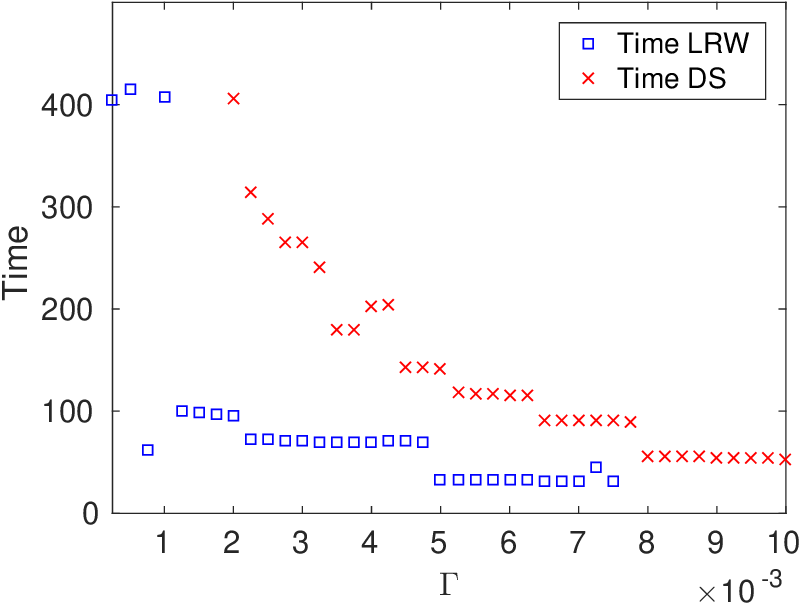}
  }
  \caption{Time of last rogue wave and permanent downshift as functions of $\Gamma$ for the a) vHONLS  and B) dvHONLS evolutions, for fixed initial data $u_0 = 0.45(1 + 0.01 \cos\mu x)$.}
  \label{time_tlrw}
\end{figure}

\section{Conclusions}
In this work we investigated the linear stability of damped Stokes waves, frequency downshifting, and rogue wave formation in a higher-order NLS model that includes viscosity.
The linear stability analysis of the damped Stokes wave solution focuses on the effects of the higher-order terms and the viscosity in two interesting regimes: the asymptotic long-time dynamics and the transition from an initial Benjamin-Feir instability to predominantly oscillatory behavior. While the long-time asymptotic results confirm and explain rigorously the findings by previous authors for similar models \cite{cartergovan, eeltink2019} (including destabilizing effects of viscosity on the higher modes and stabilization effects of diffusion), the study of the transient Benjamin-Feir instability, following in the steps of~\cite{segurstab}, is new, in that it shows that viscosity further destabilizes the initial evolution and it gives a perturbative formula for the transition time at which the initial instability begins to saturate. 

The second part of this study  combines analysis and numerics to investigate the mechanism for frequency downshifting and the time at which permanent downshift occurs.  Two spectral measures, the spectral mean $k_m$ and the spectral peak $k_{peak}$ are used to determine downshifting.  
 For the vHONLS and dvHONLS models, we obtain conditions for the spectral mean to decrease, corroborating prior  results for other dissipative HONLS models \cite{cartergovan, chb2019}.
On the other hand,  the
time  of permanent downshift in $k_{peak}$ was previously unexplored.
As a result, we determine a novel
 criterium  for predicting the occurrence and timing of permanent downshift in
$k_{peak}$ for viscous HONLS models. Our analysis explains the onset of downshifting in terms of the momentum $P$ and its first and second derivatives. We introduce the {\em precursor time $t_{P_{min}}$} to permanent downshifting as the time at which the global minimum of the linear momentum occurs and derive a perturbative estimate for its location.The estimate is in good agreement with the outcomes of an ensemble of numerical experiments. In a comparison of rogue wave activity in several damped HONLS models, we find that as many or more rogue waves develop in the vHONLS and dvHONLS models as in  the LDHONLS model.

\section{Appendix}
\label{derive_vHONLS}
In this section we provide a brief derivation of the viscous higher order nonlinear Schr\"odinger equation~\eqref{vHONLS}. Following \cite{didyza2008,cartergovan}, we start with equations that govern weakly damped two-dimensional free-surface flows:
\be\ba{rcl}
\Delta \phi &=& 0,\\
\frac{\partial\eta}{\partial t} + \frac{\partial\eta}{\partial x}\frac{\partial\phi}{\partial x} &=& \frac{\partial\phi}{\partial z} +
2\mu\frac{\partial^2 \eta}{\partial x^2} \mbox{ at } z = \eta(x,t),\\
\frac{\partial\phi}{\partial t} + \half|\nabla\phi|^2 + g\eta &=& -2\mu\frac{\partial^2\phi}{\partial z^2}  \mbox{ at } z = \eta(x,t),\\
\frac{\partial\phi}{\partial z} &\rightarrow 0& \mbox{ as }
z \rightarrow - \infty,
\ea\ee
where $\phi = \phi(x,z,t)$ is the velocity potential of the fluid, $\eta = \eta(x,t)$ is the surface displacement, and $\mu$ is the kinematic viscosity of the fluid. 

Assuming small amplitude approximations, we make the ansatz
\be
\ba{rcl}
\phi &=& \eps^2\bar\phi + \eps \left(A \e^{\ri\theta + kz} + *\right) + \cdots,\\
\eta &=& \eps^3\bar\eta + \eps \left(B \e^{\ri\theta + kz} + *\right) + \cdots,
\ea
\label{AB_exp}
\ee
where $\theta = \omega t - kx$, and where $\omega$ and $k$ are the frequency and wave number of the carrier wave. The amplitude $A$ and the mean $\bar\phi$ are functions of the slow variables $X = \eps x$, $Z = \eps z$, and $T = \eps t$, while $B$ and $\bar\eta$ are functions of $X$ and $T$ only. Here, $\eps = ka << 1$ is a measure of the steepness and $a$ is the initial surface amplitude. For small viscous effects it is assumed that $\mu$ is $O(\epsilon^2)$, so we replace it with $\eps^2\mu$, for $\mu$ an $O(1)$-term.

These assumptions yield the following equation for $A$
  \be\ba{rcl}
2i\omega\left(A_T +  \frac{\omega}{2k} A_X\right) &+& \eps\left(\left(\frac{\omega}{2k}\right)^2 A_{XX} + 4k^4|A|^2 A   +  4ik^2 \omega\mu A\right)\\
&+& \eps^2\left( - \frac{i\omega^2}{8k^3}A_{XXX} + 12\ri k^3|A|^2 A_X - 2\ri k^3A^2 A^*_x + 2k\omega A\bar{\phi}_X  - 8k\omega\mu A_X\right) = 0
\label{Eqn0}
\ea\ee
In the absence of viscosity ($\Gamma = 0$) the equation can be brought into
Hamiltonian form \cite{fd11,gramstad_trulsen_2011}. Noting that  $A^2A^*_x = \left(A^2A^*\right)_x - 2|A|^2A_x$, the term $\left(A^2A^*\right)_x$ is
eliminated  from the Hamiltonian formulation of the HONLS since it does not contribute to the Hamiltonian. The equivalent form of the vHONLS equation is
\be\ba{rcl}
2i\omega\left(A_T + \frac{\omega}{2k} A_X\right) &+& \eps\left(\left(\frac{\omega}{2k}\right)^2 A_{XX} + 4k^4|A|^2 A  +  4ik^2 \omega\mu A\right)\\
&+& \eps^2\left( -\frac{i\omega^2}{8k^3}A_{XXX} + 16ik^3|A|^2 A_X + 2k\omega A\bar{\phi}_X  - 8k\omega\mu A_X\right) = 0.
\label{Eqn1}
\ea\ee

Introducing the following dimensionless and translated variables
  \be
\tau = \frac{\eps\omega}{8}T,\
\chi = \frac{\omega T}{2} - k X,\
\zeta = kZ,\
u = 2\sqrt{2}\frac{k^2}{\omega}A,\
\tilde\phi = \frac{2k^2}{\omega}\bar\phi,\
\Gamma = \frac{16k^2\mu}{\omega},
 \ee
yields 
  \[
A_T = \omega\left(\frac{1}{8}\eps\omega A_\tau + \frac{\omega}{2} A_\chi \right),
\quad A_X = - k A_\chi,\quad
A_{XX} = k^2 A_{\chi\chi},\quad
A_{XXX} = -k^3 A_{\chi\chi\chi},
\]
which in turn gives
\[\ba{rcl}
2i\omega\left(A_T + \frac{\omega}{2k} A_X\right) &=& \frac{i\eps\omega^3}{8\sqrt 2 k^2} u_\tau, \\
& & \\
\left(\frac{\omega}{2k}\right)^2 A_{XX} + 4k^4|A|^2 A + 4ik^2 \omega\mu A
&=& \frac{\omega^3}{8\sqrt 2 k^2}\left(u_{\chi\chi}  + 2|u|^2 u + i\Gamma u\right),\\
& & \\
-\frac{i\omega^2}{8k^3}A_{XXX} + 16ik^3|A|^2 A_X + 2k\omega A\bar{\phi}_X - 8k\omega\mu A_X
&=& \frac{\omega^3}{8\sqrt 2 k^2}\left( \frac{i}{2} u_{\chi\chi\chi} - 8i|u|^2 u_\chi - 4 u\tilde{\phi}_\chi + 2\Gamma  u_\chi\right).
\ea\]

After solving Laplace's equation on $-\infty < \zeta < 0$ for $\tilde\phi$ in terms of $u$, with $\tilde\phi_\zeta = \half\left(|u|^2\right)_\chi$ at $\zeta = 0$ and $\tilde\phi_\zeta \rightarrow 0$ as $\zeta \rightarrow -\infty$, we arrive at the equation:
\be
{\ri}u_\tau + \left(u_{\chi\chi}  + 2|u|^2 u + \ri\Gamma u\right) + \eps\left( \frac{\ri}{2} u_{\chi\chi\chi} - 8{\ri}|u|^2 u_\chi + 2u\left[H\left(|u|^2\right)\right]_\chi + 2\Gamma  u_\chi\right) = 0.
\ee

\section*{Acknowledgements}
We thank the referees for their insightful comments and for pointing at several relevant references.
We gratefully acknowledge support by Summer Undergraduate Research with Faculty awards at
the College of Charleston and University of Central Florida, and by the School of Science and Mathematics at the College of Charleston.
A.~Calini and C.~M.~Schober would like to thank the Isaac Newton Institute for Mathematical Sciences for support and hospitality during the programme {\em Dispersive hydrodynamics: mathematics, simulation and experiments, with applications in nonlinear waves},  when part of the work on this paper was undertaken. This work was partially supported by: EPSRC grant number EP/R014604/1, and by the Simons Foundation through award \#527565 (PI: C.M.~Schober).

\bibliography{VMbib} 

\begin{thebibliography}{10}

\bibitem{ahhs01}
M.J. Ablowitz, J.~Hammack, D.~Henderson, and C.M. Schober.
\newblock Long-time dynamics of the modulational instability of deep water
  waves.
\newblock {\em Physica D: Nonlinear Phenomena}, 152-153:416--433, 2001.
\newblock Advances in Nonlinear Mathematics and Science: A Special Issue to
  Honor Vladimir Zakharov.

\bibitem{AS}
M.~Abramowitz and I.A. Stegun.
\newblock {\em Handbook of Mathematical Functions with Formulas, Graphs, and
  Mathematical Tables}.
\newblock Dover, New York City, ninth dover printing, tenth gpo printing
  edition, 1964.

\bibitem{barrett}
J.H. Barrett.
\newblock Second order complex differential equations with a real independent
  variable.
\newblock {\em Pacific Journal of Mathematics}, 8(2):187–200, 1958.

\bibitem{bellman}
R.~Bellman.
\newblock {\em Stability Theory of Differential Equations}.
\newblock Mc Graw-Hill, 1953.

\bibitem{cs02}
A.~Calini and C.~M. Schober.
\newblock Homoclinic chaos increases the likelihood of rogue wave formation.
\newblock {\em Phys. Lett. A}, 298:335--349, 2002.

\bibitem{cartergovan}
J.D. Carter and A.~Govan.
\newblock Frequency downshift in a viscous fluid.
\newblock {\em European Journal of Mechanics - B/Fluids}, 59:177–185, 2016.

\bibitem{chb2019}
J.D. Carter, D.~Henderson, and I.~Butterfield.
\newblock A comparison of frequency downshift models of wave trains on deep
  water.
\newblock {\em Physics of Fluids}, 31(1):013103, 2019.

\bibitem{didyza2008}
F.~Dias, A.I. Dyachenko, and V.E. Zakharov.
\newblock Theory of weakly damped free-surface flows: A new formulation based
  on potential flow solutions.
\newblock {\em Phys. Let. A}, 372:1297--1302, 2008.

\bibitem{dysthe}
K.~B. Dysthe.
\newblock Note on modification to the nonlinear schr\"odinger equation for
  application to deep water waves.
\newblock {\em Proc. Roy. Soc. Lond. A}, 369:105, 1979.

\bibitem{eeltink2019}
D.~Eeltink, A.~Armaroli, Y.~M. Ducimeti\`ere, J.~Kasparian, and M.~Brunetti.
\newblock Single-spectrum prediction of kurtosis of water waves in a
  nonconservative model.
\newblock {\em Phys. Rev. E}, 100:013102, Jul 2019.

\bibitem{eeltink2017}
D.~Eeltink, A.~Lemoine, H.~Branger, O.~Kimmoun, C.~Kharif, J.~D. Carter,
  A.~Chabchoub, M.~Brunetti, and J.~Kasparian.
\newblock Spectral up-- and downshifting of akhmediev breathers under wind
  forcing.
\newblock {\em Physics of Fluids}, 29:107103, 10 2017.

\bibitem{fd11}
D.~S.~Dutykh F.~Fedele.
\newblock Hamiltonian form and solitary waves of the spatial dysthe equations.
\newblock {\em JETP Letters}, 94:840--844, 2011.

\bibitem{gramstad_trulsen_2011}
O.~Gramstad and K.~Trulsen.
\newblock Hamiltonian form of the modified nonlinear schr\"odinger equation for
  gravity waves on arbitrary depth.
\newblock {\em Journal of Fluid Mechanics}, 670:404–426, 2011.

\bibitem{hara_mei_1991}
T.~Hara and C.C. Mei.
\newblock Frequency downshift in narrowbanded surface waves under the influence
  of wind.
\newblock {\em Journal of Fluid Mechanics}, 230:429–477, 1991.

\bibitem{hille1976ordinary}
E.~Hille.
\newblock {\em Ordinary Differential Equations in the Complex Domain}.
\newblock A Wiley-Interscience publication. Wiley, 1976.

\bibitem{huang96}
N.E. Huang, S.R. Long, and Z.~Shen.
\newblock The mechanism for frequency downshift in nonlinear wave evolution.
\newblock {\em Advances in Applied Mechanics}, 32:59--117C, 1996.

\bibitem{islas}
A.~Islas and C.M. Schober.
\newblock Rogue waves and downshifting in the presence of damping.
\newblock {\em Nat. Hazards Earth Syst. Sci.}, 11:383–399, 2011.

\bibitem{KatoOikawa95}
Y.~Kato and M.~Oikawa.
\newblock Wave number downshift in modulated wavetrain through a nonlinear
  damping effect.
\newblock {\em Journal of the Physical Society of Japan}, 64(12):4660--4669,
  1995.

\bibitem{kmwy09}
A.~Q.~M. Khaliq, J.~Mart\'{i}n-Vaquero, B.~A. Wade, and M.~Yousuf.
\newblock Smoothing schemes for reaction-diffusion systems with nonsmooth data.
\newblock {\em J. Comp. Appl. Math.}, 223:374--386, 2009.

\bibitem{lake_yuen_rungaldier_ferguson_1977}
B.M. Lake, H.C. Yuen, H.~Rungaldier, and W.E. Ferguson.
\newblock Nonlinear deep-water waves: theory and experiment. part 2. evolution
  of a continuous wave train.
\newblock {\em Journal of Fluid Mechanics}, 83(1):49–74, 1977.

\bibitem{lkx15}
X.~Liang, A.Q.M. Khaliq, and Y.~Xing.
\newblock Fourth order exponential time differencing method with local
  discontinuous galerkin approximation for coupled nonlinear schrödinger
  equations.
\newblock {\em Communications in Computational Physics}, 17:510--541, 2015.

\bibitem{lo_mei_1985}
E.~Lo and C.C. Mei.
\newblock A numerical study of water-wave modulation based on a higher-order
  nonlinear schrödinger equation.
\newblock {\em Journal of Fluid Mechanics}, 150:395–416, 1985.

\bibitem{melville_1982}
W.K. Melville.
\newblock The instability and breaking of deep-water waves.
\newblock {\em Journal of Fluid Mechanics}, 115:165–185, 1982.

\bibitem{okamura96}
M.~Okamura.
\newblock Long time evolution of standing gravity waves in deep water.
\newblock {\em Wave Motion}, 23:279--287, 1996.

\bibitem{OLBC}
F.W.J. Olver, D.W. Lozier, R.F. Boisvert, and C.W. Clark.
\newblock {\em NIST handbook of mathematical functions}.
\newblock Cambridge University Press, 2010.

\bibitem{SS}
C.M. Schober and M.~Strawn.
\newblock The effects of wind and nonlinear damping on rogue waves and
  permanent downshift.
\newblock {\em Elsevier}, Physica D:81--98, 2015.

\bibitem{segurstab}
H.~Segur, D.~Henderson, J.~Carter, J.~Hammack, D.~Pheiff, and K.~Socha.
\newblock Stabilizing the benjamin–feir instability.
\newblock {\em J. Fluid Mech.}, 539:229–271, 2005.

\bibitem{su82}
M.Y. Su.
\newblock Evolution of groups of gravity waves with moderate to high steepness.
\newblock {\em The Physics of Fluids}, 25(12):2167--2174, 12 1982.

\bibitem{trulsen97}
EP. Tood, editor.
\newblock {\em Freak waves --a three dimensional wave simulation}. Nat. Acad.
  Press, 1997.

\bibitem{TrulsenDysthe1990}
K.~Trulsen and K.B. Dysthe.
\newblock {\em Frequency Down-Shift Through Self Modulation and Breaking},
  volume 178 of {\em NATO ASI Series}, pages 561--572.
\newblock Springer Netherlands, Dordrecht, 1990.

\bibitem{tulin99}
M.P. Tulin and T.~Waseda.
\newblock Laboratory observations of wave group evolution, including breaking
  effects.
\newblock {\em J. Fluid Mech.}, 378:197--232, 1999.

\bibitem{uchiyama}
Y.~Uchiyama and T.~Kawahara.
\newblock A possible mechnism for frequency downshift in nonlinear wave
  modulation.
\newblock {\em Wave Motion}, 20:99--110, 1994.

\bibitem{zaugcarter21}
C.R. Zaug and J.D. Carter.
\newblock Dissipative models of swell propagation across the pacific.
\newblock {\em Studies in Applied Mathematics}, 147(4):1519--1537, 2021.

\end{thebibliography}
\bibliographystyle{plain}

\end{document}